\documentclass[amsmath,  floatfix, nofootinbib, prx, aps, usenames, dvipsnames,10pt, allowfontchageintitle, a4paper]{quantumarticle}

\usepackage{hyperref} %quantumarticle.cls takes care of colours
\usepackage{jabbrv}

\let\olddesc\description
\def\description{\olddesc\setlist[itemize]{leftmargin=*,labelindent=-12pt}}

\usepackage{lineno}

\usepackage{xr}

\usepackage[english]{babel}

\usepackage{graphicx}
\usepackage{caption}
\usepackage{subcaption}
\usepackage[export]{adjustbox}

\usepackage{xcolor}

\usepackage{tcolorbox}
\tcbuselibrary{breakable}
\tcbuselibrary{skins}
\tcbset{enhanced, colback=white, arc=0mm, outer arc=0mm, fonttitle=\bfseries, fontupper=\small}
\newtcolorbox[auto counter]{examplebox}[2][]{%
title=Box~\thetcbcounter: #2,#1}

\usepackage{amsmath}
\usepackage{amsthm}
\usepackage{amssymb}
\usepackage{mathtools}
\usepackage{thmtools}
\usepackage{tensor}

\usepackage{bm}  % Define \bm{} to use bold math fonts
\usepackage{bbm}
\usepackage{dsfont}
\usepackage{cases}
\usepackage{nicefrac} 
\usepackage{xfrac}

\usepackage[shortlabels]{enumitem}

\usepackage{multicol}

\newtheorem{theorem}{Theorem}%[section]
\newtheorem{result}{Main result}

\newtheorem{corollary}{Corollary}
\newtheorem{proposition}[corollary]{Proposition}
\newtheorem{lemma}[corollary]{Lemma}
\newtheorem{definition}[corollary]{Definition}
\newcommand{\norm}[1]{\left\lVert#1\right\rVert}

\newcommand{\ket}[1]{\left| #1 \right\rangle} % for Dirac bras
\newcommand{\bra}[1]{\left\langle #1 \right|} % for Dirac kets
\newcommand{\braket}[2]{\left\langle#1 | #2\right\rangle}
\newcommand{\ketbra}[2]{|#1\rangle\!\langle#2|}
\newcommand{\proj}[1]{|#1\rangle\!\langle#1|}

\newcommand{\rank}{\mathrm{rank}}

\newcommand{\tr}{\mathrm{tr}}
\newcommand{\Tr}{\tr}
\newcommand{\e}{\mathrm{e}}
\newcommand{\1}{\mathbbm{1}}
\newcommand{\mf}{\mathfrak}
\newcommand{\mc}{\mathcal}
\newcommand{\s}{S}
\graphicspath{{img/}{../img/}}

\begin{document}

\title{Extensive R\'enyi entropies in matrix product states}

\author{Alberto Rolandi}
\email{rolandia@student.ethz.ch}
\affiliation{Institute for Theoretical Physics, ETH Z{\"u}rich, 8093 Z{\"u}rich, Switzerland}

\author{Henrik Wilming}
\email{henrikw@phys.ethz.ch}
\affiliation{Institute for Theoretical Physics, ETH Z{\"u}rich, 8093 Z{\"u}rich, Switzerland}

\date{}

\begin{abstract} 
\noindent 
	We prove that all R\'enyi entanglement entropies of spin-chains described by generic (gapped), translational invariant matrix product states (MPS) are extensive for disconnected sub-systems: All R\'enyi entanglement entropy densities of the sub-system consisting of every $k$-th spin are non-zero in the thermodynamic limit if and only if the state does not converge to a product state in the thermodynamic limit. 
	Furthermore, we provide explicit lower bounds to the entanglement entropy in terms of the expansion coefficient of the transfer operator of the MPS and spectral properties of its fixed point in canonical form. 
	As side-result we obtain a lower bound for the expansion coefficient and singular value distribution of a primitve quantum channel in terms of its Kraus-rank and entropic properties of its fixed-point. For unital quantum channels this yields a very simple lower bound on the distribution of singular values and the expansion coefficient in terms of the Kraus-rank. 
Physically, our results are motivated by questions about equilibration in many-body localized systems, which we review. 
\end{abstract}

\maketitle

%%%% Main text %%%%%

\section{Introduction and main results}
Matrix product states (MPS), also known as finitely correlated states \cite{Fannes1992}, and tensor networks more generally, have become indispensable tools to study quantum many-body systems (see, for example, \cite{Bridgeman2016} for a thorough introduction). 
The essential reason for this is that matrix product states provide a variational class of many-body wavefunctions that is both treatable in a computationally efficient manner and at the same time captures
the relevant correlation structure of ground states of gapped local Hamiltonians, i.e., non-critical systems.  
The two crucial features that they capture are i) the Area Law of entanglement and ii) the exponential decay of correlations, which hold for any ground state of a gapped Hamiltonian in one spatial dimension \cite{Hastings2004,Hastings2006,Hastings2007}. 
The Area Law states that the entanglement entropy of a subsystem $A$ fulfills
\begin{align}\label{eq:arealaw}
	S(\rho_A) \leq c |\partial A|,
\end{align}
where $|\partial A|$ denotes the size of the boundary of $A$,  $c$ is some constant, and $S$ denotes von~Neumann entropy. 
MPS with fixed bond-dimension $D$ (see below for a definition of bond-dimension) more generally fulfill an Area Law not just for the von~Neumann entropy, but for all R\'enyi entropies $S_\alpha$:
\begin{align}
	S_\alpha(\rho_A) \leq |\partial A|\log(D).
\end{align}
 
 The exponential decay of correlations on the other hand means that connected corelation functions of local operators $A$ and $B$ decay exponentially with the distance $d(A,B)$ between the support of the operators:
 \begin{align}\label{eq:clustering}
	 \big| \langle AB\rangle - \langle A\rangle\langle B\rangle\big| \leq \norm{A}\norm{B} \e^{- d(A,B)/\xi},
\end{align}
where $\xi\geq 0$ is the \emph{correlation length} and $\langle \cdot \rangle$ denotes an expectation value with respect to the quantum state in question.  
Indeed, MPS with fixed bond-dimension either fulfill \eqref{eq:clustering} with $\xi>0$ or show infinite-ranged correlations for certain observables $A,B$. 
That is, MPS with fixed bond-dimension cannot represent quantum sates that show a power-law decay of correlations. 

In this paper, we show that within the set of MPS with finite correlation length (so-called \emph{gapped} MPS), a similar dichotomy shows up for the R\'enyi entanglement entropies of generic translational invariant matrix-product states, when we consider \emph{disconnected} sub-systems consisting of every $k$-th spin of the total system: either all R\'enyi entanglement entropies grow proportional to $n/k$ as the system-size $n$ is increased, or the matrix product state converges to a product state (in a sense made precise below). 
In the latter case, all R\'enyi entropies with $\alpha>1$ converge to zero. Therefore, for these entropies it is not possible to have an intermediate kind of behaviour, such as a growth of the order $\sqrt{n/k}$. To state our result, we mention here already that a central ingredient in the study of matrix product states is the \emph{transfer operator} $T$, which induces a completely positive map $\mc T$ with the same eigenvalues and singular values as $T$. We say that the transfer operator is \emph{gapped} if it has a unique, non-degenerate eigenvalue of modulus $1$.
\begin{figure*}
	\centering
	\includegraphics[width=14cm]{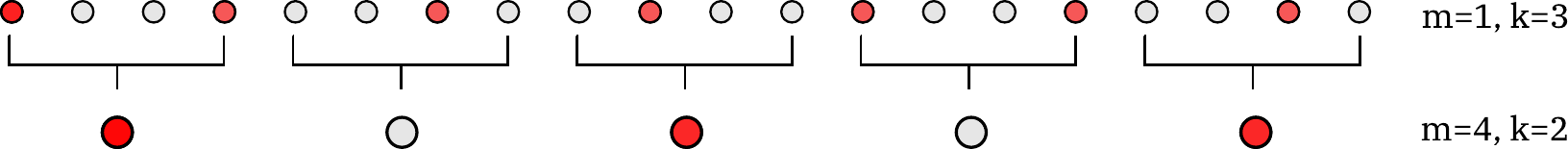}
	\caption{We can re-group the sites of the chain into blocks of $m$ spins and think of each block as a "super-spin". Our results then directly transfer to studying the entropy per block. Technically, this re-grouping simply amounts to taking the $m$-th power of the transfer operator. }
	\label{fig:renormalization}
\end{figure*}
\begin{result}[Extensivity of R\'enyi entropies, informal]\label{thm:main_informal} Consider a system described by a translational invariant matrix product state with gapped transfer operator and fixed bond dimension. Denote by $\rho_k$ the density matrix of the sub-system consisting of every $k$-th spin. Then its R\'enyi entanglement entropies are proportional to the system-size $n$:
	\begin{align}
		\lim_{n\rightarrow \infty} \frac{k}{n} S_\alpha(\rho_k) \geq s_\alpha^{(k)},
	\end{align}
	 and $s_\alpha^{(k)}=0$ \emph{if and only if} $\rho_k$ approaches a product-state in the thermodynamic limit, in which case even the total R\'enyi entropies with $\alpha>1$ vanish in the thermodynamic limit. 
\end{result}
One might at first sight think that this result violates the Area Law of MPS. But this is not the case, for if $A$ consists of every $k$-th spin, we have $|\partial A|=2|A|=2n/k$. 
The property of having a gapped transfer operator is generic, in the sense that if we choose the tensor describing the translational invariant MPS at random from a reasonable ensemble, then we obtain a gapped transfer operator with probability one.
We have expressed our result for the subsystem consisting of every $k$-th spin, but we can also instead partition the system into blocks of $m$ consecutive spins and consider each of the blocks as one "super-spin" (see Fig.~\ref{fig:renormalization}). 
Our result then directly applies to the super-spins as well, since the transfer operator describing the system in terms of super-spins is simply the $m$-th power $T^m$ of the original transfer operator.   

The lower bound $s^{(k)}_\alpha$ to the entropy density is determined by the largest eigenvalue of what we call the \emph{twisted transfer operator} below and can easily be computed numerically given the tensor describing the MPS. 
Nevertheless, it is of interest to have quantitative estimates of this quantity.
We therefore derive lower bounds to $s^{(k)}_\alpha$ in terms of properties of the transfer operator itself.
In deriving these, we obtain some general side-results about completely positive (CP) maps, which we believe to be of independent interest. 

As before, we assume that the transfer operator has a spectral gap $\Delta=1-|t_2|>0$, where $t_i$ denote the eigenvalues of the transfer operator, ordered with decreasing absolute value, and $t_1=1$. 
Then the adjoint $\mc T^*$ of the associated completely positive map $\mc T$ has a unique fixed-point density matrix $\Lambda$. A unital or trace-preserving completely positive map is called \emph{primitive} if both $\mc T$ and $\mc T^*$ have a unique eigenvalue of modulus $1$ and the associated fixed-points are strictly positive \cite{Wolf2012}. 
We also introduce the \emph{singular gap} $\Delta_s$, given by
\begin{align}
	\Delta_s = s_1(T) - s_2(T),
\end{align}
where $s_i(T)=s_i(\mc T)$ denote the singular values of $T$, i.e., the eigenvalues of $\sqrt{T^\dagger T}$.
The largest singular value of a unital CP-map $\mc T$ fulfills $s_1(\mc T)\geq 1$ and 
$s_1(\mc T)=1$ if and only if $\mc T$ is also trace-preserving \cite{PerezGarcia2006}. 

Let us denote by $\lambda_{\min{}}$ the smallest non-zero eigenvalue of $\Lambda$ and recall that $D$ is the bond-dimension.

We summarize some of our bounds in terms of the following main result, which provides a lower bound to $s_\alpha^{(k)}$ in terms of $\lambda_{\min{}}$ and the singular values of $T$, in the limit $k\rightarrow \infty$. 
To state it, we denote by $\mf s$ the \emph{swap-operator} that acts as $\mf s\ket{a}\otimes\ket{b} = \ket{b}\otimes \ket{a}$. 
For simplicity, we here also assume that $\Lambda>0$ (in which case $\mc T$ is primitive and unital), but provide results for the general case as well as more refined versions of the bounds below in section~\ref{sec:lower-bounds}.  In particular, these allow to replace much of the dependence on the singular values $s_i(\mc T)$ with entropic properties of $\Lambda$ in the following result. 
\begin{result}[Lower bound to entropy, informal] Under the same conditions as in result~\ref{thm:main_informal} and if $\Lambda>0$, we have
	\begin{align}
		\lim_{k\rightarrow \infty} s_2^{(k)} &= -\log\left(\tr\left[\Lambda\otimes\Lambda \mf s \mc T \otimes\mc T [\mf s]\right]\right)\\
		&\geq -\log(1 - \lambda_{\min{}}^2[D^2-\sum_i s_i^2(T)]) \geq 0,
	\end{align}
	with
%	\begin{align}
		$\sum_{i=1}^{D^2} s_i^2(T) \leq D^2$.
%	\end{align}
	The convergence is exponential in $k$ and the bound becomes trivial if and only if $D=1$ (and hence the state is given by a product state). 
	Moreover, if $\Lambda$ is maximally mixed, then
	\begin{align}
		\lim_{k\rightarrow \infty} s_2^{(k)} &\geq -\log(1/D^2 + (1-\Delta_s)^2),
	\end{align}
	and $0\leq \Delta_s \leq 1$. 
	In the limit of large block-sizes, the entropy per block converges (exponentially quickly) to
	\begin{align}
		\lim_{m\rightarrow \infty} \lim_{k\rightarrow \infty} s_\alpha^{(k)} = 2 S_\alpha(\Lambda).
	\end{align}
\end{result}
While the lower bounds given above for finite $k$ on first sight only seems to deliver information for the R\'enyi entropy $S_2$, they actually provide a lower bound for all R\'enyi entropies, since $S_\alpha\geq \frac{1}{2}S_2$ for all $\alpha$.

Stated as above, our result is one on MPS and their entanglement properties. However, given any quantum channel $\mc C$ of Kraus-rank $d$, we can construct from it a translational invariant MPS on a spin chain with local Hilbert-space dimension $d$ whose transfer operator matches $\mc C$: $\mc C=\mc T^*$. 
But we know that the entropy density $\frac{k}{n} S_\alpha$ of a system of $d$-dimensional constituents is always upper bounded by $\log(d)$.
The above bounds then lead us to results bounding how quickly a primitive quantum channel $\mc T^*$ mixes towards its fixed-point density matrix $\Lambda$ in terms of its Kraus-rank $d$.
To make the results easier to interpret, we introduce the \emph{expansion coefficient $\kappa(\mc T)$} of a unital or trace preserving, primitive CP-map $\mc T$:
\begin{align}
	\kappa(\mc T) := \sup_{X:\norm{X}_2=1} \norm{\mc T[X]- \lim_{n\rightarrow\infty}\mc T^n[X]}_2, 
\end{align}
where $\norm{X}_2=\sqrt{\Tr[X^\dagger X]}$ denotes the Schatten 2-norm induced by the Hilbert-Schmidt inner product. 
Equivalently, if $\mc T$ is unital and writing $\mc P_\Lambda[X]=\1\Tr[\Lambda X]$, we have $\kappa(\mc T) = \norm{\mc T-\mc P_\Lambda}$ when $\mc T$ is viewed as super-operator on the Hilbert-space of $D\times D$ matrices with Hilbert-Schmidt inner product. This implies $\kappa(\mc T)=\kappa(\mc T^*)$.
The expression \emph{expansion coefficient} is borrowed from the literature of \emph{quantum expanders} \cite{Hastings2007b,Hastings2007a,Ben-Aroya2010,Gross2008,Harrow2008}, where the term is used for the special case of unital and trace-preserving CP-maps (unital quantum channels). It measures how quickly an arbitrary input to the CP-map approaches the fixed-point. 
Indeed, suppose that $\kappa(\mc T^*)\leq \epsilon/\sqrt{D}$. Then for any density matrix $\rho$, we have
\begin{align}
	\norm{\mc T^*[\rho]-\Lambda}_1 \leq \sqrt{D}\norm{\mc T^*[\rho]-\Lambda}_2 \leq \epsilon. 
\end{align}
While we later derive a general bound relating $\kappa(\mc T)$ to the Kraus-rank and the entropic properties of the fixed-point $\Lambda$, we here state a simplified version for the special case of unital channels. 
\begin{corollary}\label{cor:1}
	Let $\mc T$ be a primitive, unital, trace-preserving, completely positive map on a $D$-dimensional Hilbert-space with Kraus-rank $d$. 
	Then
	\begin{align}
	\frac{1}{d} \leq \frac{\sum_i s_i(\mc T)^2}{D^2} \leq \frac{1}{D^2} + \kappa(\mc T)^2.
	\end{align}
	In particular, its expansion coefficient is bounded as
	\begin{align}
		\label{eq:boundkappa}
	\kappa(\mc T)^2 \geq \frac{1}{d} - \frac{1}{D^2} \geq 0 .
	\end{align}
\end{corollary}
At this point it is interesting to know that for any $n\in \mathbb N$, there exist primitive, unital and trace-preserving CP-maps with $D=2^n$ and $d=2$ \cite{Hanson2020,Daniel}. 
The inequality then tells us that such channels fulfill
\begin{align}
	\kappa(\mc T)^2 \geq \frac{1}{2}\left(1- 2^{-(2n-1)}\right).
\end{align}
In fact, for primitive, unital and trace-preserving CP-maps the expansion coefficient is given by $1-\Delta_s$. The bound then translates to 
\begin{align}\label{eq:boundsingulargap}
\Delta_s\leq 	\Delta_s(2-\Delta_s) \leq 1 + 1/D^2 - 1/d\leq 1,
\end{align}
where the first inequality follows from $\Delta_s\leq 1$ for such maps.
A particular simple application of the bound is the following: Consider the CP-map $\mc T[A] = \1 \Tr[A]/D = \mc P_{\1/D}[A]$. Both $\mc T$ and $\mc T^*$ map every density matrix to the maximally mixed state. It has $\kappa(\mc T)=0$. The bound given in \eqref{eq:boundkappa} then implies the well-known fact that this CP-map requires maximal Kraus-rank. It is hence tight in this limit.  

One may wonder whether a similar result also holds for the \emph{spectral gap} $\Delta$, which is always at least as large as the singular gap. 
Here, even the weak inequality $\Delta \leq 1+1/D^2-1/d$ can be violated when the singular gap is replaced by the spectral gap. 
An example for $D=2$ and $d=2$ is given by the mixed-unitary map of the form
\begin{align}
	\mc T'[\cdot] := \frac{1}{2} U_1 \cdot U_1^\dagger + \frac{1}{2} U_2 \cdot U_2^\dagger,
\end{align}
with $U_1=\sigma_y$ and $U_2 = \frac{1}{\sqrt{2}}\big (\1 + \mathrm i/\sqrt{2}(\sigma_x + \sigma_z)\big)$. This map is primitive, but has $\Delta_s=0$ and $\Delta=1$. Indeed, it has two non-zero singular values, both equal to $1$. The first bound in corollary~\ref{cor:1} then shows $d \geq 2$. It is therefore tight, while the second bound only gives the trivial bound $d\geq 4/5$. 

To conclude the introduction and overview of our results, let us mention that using $\mc T'$ to construct an MPS with bond dimension $2$ on a spin-1/2 chain yields a state with the following interesting properties: a) every connected subsystem of size $l$ with $2\leq l\leq n-2$ has R\'enyi-2 entropy $S_2 = 2\log(2)$ saturating the Area Law bound, b) every individual spin is maximally mixed and c) whenever two subsystems are separated by at least two spins, they are uncorrelated. In particular, this implies that the sub-system consisting of every $k$-th spin is maximally mixed for every $k\geq 3$, despite the fact that the global state is pure. All these properties can be derived from $\tr[U_1 U_2^\dagger]=0$ and $\mc T'\circ \mc T' = \mc P_{\1/2}$ alone, as shown in appendix~\ref{app:exampleMPS}.

\subsection{Structure of Paper}
We first discuss the physical motivation for our results, which comes from the problem of equilibration of many-body localized systems (readers mainly interested in our technical results may skip this section). 
We then briefly review the formalism of matrix product states and show how R\'enyi entropies are calculated in this framework using the the twisted transfer operator.  
In section~\ref{sec:translational-invariant}, we provide the technical version and the proof for our first main-result, showing that the R\'enyi entropies are extensive.
Section~\ref{sec:lower-bounds} provides the technical formulation on our lower-bounds for the entropy density and their corollary for quantum channels. 
We end with a brief conclusion, where we discuss some open problems.  
The construction and properties of the MPS mentioned at the end of the introduction are discussed in appendix~\ref{app:exampleMPS}. 
To ease readability, some purely technical proofs of Lemmata are delegated to appendix~\ref{app:lemmas}.

\section{Physical motivation for our results} 
\label{sec:physical-motivation}
The motivation for our results comes from the study of equilibration of many-body localized systems \cite{Nandkishore2015}.
When a complex quantum system is initialized in a simple initial state and allowed to evolve under its natural unitary time-evolution, sub-systems will generically become entangled with the remainder of the system and loose memory of their precise initial conditions \cite{Gogolin2016}. 
More precisely, the quantum state of the sub-system will show small fluctuations around a stationary state, interspersed with rare, larger deviations and a recurrence to the initial state after a time double-exponentially large in the system-size. 
In this case, we say that the system \emph{equilibrates}. 
Generically, even more is true, namely that the stationary state of the small sub-system may be computed from a thermal, statistical ensemble of the full system, even though global properties of the full system at no time can be described by this ensemble. In this case, we say that the system also \emph{thermalizes}. 
While this general picture of equilibration and thermalization is agreed upon, see for example the review \cite{Gogolin2016}, a detailed understanding of when, how and after how much time equilibration and thermalization actually happens is still missing. 
In particular, no general criterion has been found that allows to decide whether a presented many-body Hamiltonian acting on a particular initial state will lead to equilibration and thermalization. 
Furthermore, even though rigorous arguments exist that show equilibration under fairly general (although often not explicitly checkable) conditions \cite{Tasaki1998,Popescu2006,Reimann2008,Linden2009,Goldstein2010,Short2011,Reimann2012,Reimann2012a,Short2012,Masanes2013,Gogolin2016,Gallego2017}, no general prediction can be made for how long it will take for a particular observable to equilibrate (see, however, Refs.~\cite{Linden2010,Short2012,Goldstein2013,Malabarba2014,Goldstein2015,Garcia-Pintos2017,Wilming2017,DeOliveira2017,Dabelow2020}).   
Indeed, recent years have shown several surprises, such as the existence of many-body localized systems, which equilibrate but don't thermalize due to lack of transport, or many-body scared systems \cite{Bernien2017,Shiraishi2017,Turner2018,Moudgalya2018}, some of which show perfect oscillations of all, even non-local, observables for certain simple initial states, despite being non-integrable (see, for example, Refs.~\cite{choi2019emergent,schecter2019weak,Chattopadhyay2020,Iadecola2020}). 
These findings show that the equilibration and thermalization behaviour of complex quantum systems is a complex and fascinating field of research.

One of the central emerging insights in this area is that the entanglement structure of high-energy eigenstates of the many-body system are important for their equilibration behaviour. 
We emphasize here already that this insight is logically distinct from the famous \emph{eigenstate thermalization hypothesis} \cite{Srednicki1994,Srednicki1999,DAlessio2016,Gogolin2016}, which explains \emph{thermalization} due to the entanglement property of energy eigenstates, but has to assume that the system equilibrates in the first place.   
Indeed, recently, it was shown in Ref.~\cite{Wilming2019} that a generic, interacting many-body system initialized in a low-entangled state equilibrates to exponentially good precision in the system-size provided that the energy eigenstates are sufficiently entangled. 
Here, sufficiently entangled means that for any energy eigenstate $\ket{E}$ with energy $E$, the \emph{R\'enyi entanglement entropies} $S_\alpha$ of \emph{some} suitable sub-system $A(E)$ is extensive:
\begin{align} \label{eq:ergodicity}
	S_\alpha(\rho_{A(E)}) \geq g (E/n)\, n,
\end{align}
where $n$ is the system-size, $g$ is a sufficiently regular function of the energy-density and $\rho_{A(E)}$ denotes the reduced density matrix of $\ket{E}$ on the sub-system $A(E)$. 
Systems that fulfill this property were dubbed \emph{entanglement-ergodic} in Ref.~\cite{Wilming2019}. 
Thus, entanglement-ergodic systems equilibrate exponentially well. Conversely, it can be shown that systems with perfect many-body scars, that show periodic revivals of the initial states cannot be entanglement-ergodic \cite{Alhambra2020}. Indeed, in such systems there always are energy eigenstates whose entanglement entropies are bounded by $\log(n)$ for \emph{arbitrary} sub-systems. 
Entanglement-ergodicity therefore seems to provide an interesting structural insight connecting the entanglement properties of energy eigenstates with the equilibration behaviour of complex quantum systems. 
However, at first sight, it may seem to fail to explain why many-body localized (MBL) systems equilibrate. This is because it is known that energy eigenstates of MBL systems fulfill Area Laws \cite{Bauer2013} and are well described by MPS \cite{Friesdorf2015}. 
At first sight, this suggests that the eigenstates of MBL systems feature much less entanglement than required for entanglement-ergodicity. 
However, entanglement-ergodicity allows to choose the region $A(E)$ arbitrarily. By choosing a region $A(E)$ to consist of every $k$-th spin, the results in this paper make it plausible that even MBL systems are entanglement-ergodic (while we only show results for translationally invariant MPS, we expect that disordered MPS typically have even higher entropies).

\section{Matrix-product states, R\'enyi entropies and the twisted transfer operator}
\label{sec:background}
Let us first briefly review the formalism of matrix product states. For an in-depth introduction to tensor networks see for example \cite{Bridgeman2016}. We consider a chain of $n$ spins, each of which is described by a $d$-dimensional Hilbert-space $\mathcal{H}_S$. In the following we refer to $d$ as the \emph{physical dimension}. An arbitrary state vector $\ket{\psi}$ can be written as
\begin{equation}\label{arb_state}
	\ket{\psi} = \sum_{i_1,...,i_n = 0}^{d-1} \psi^{i_1,...,i_n}\ket{i_1,...,i_n}
\end{equation}
for some $d^n$ parameters $\psi^{i_1,...,i_n}$. One can think of $\ket{\psi}$ as a $n$-$0$ tensor and of $\psi^{i_1,...,i_n}$ as its components. 
A matrix product state of bond dimension $D$ is a pure state such that
\begin{align}\label{MPS_def}
	\psi^{i_1,...,i_n} &= \tr\left[(A^{(1)})^{i_1}(A^{(2)})^{i_2}...(A^{(n)})^{i_n} B\right],
\end{align}
where both $B$ and each $(A^{(j)})^{i_k}$ are $D\times D$-matrices. 
The boundary operator $B$ specifies the boundary condition of the MPS. 
In particular, if $B=\1$ we have periodic boundary conditions and if $B=\ketbra{\Omega}{\Omega'}$ for some vectors $\ket{\Omega},\ket{\Omega'}$ we have open boundary conditions. In the following, we will only consider periodic boundary conditions.

The $D$-dimensional Hilbert-space $\mathcal{H}_B\simeq \mathbb C^D$ is generally referred to as the \textit{bond Hilbert space} and $D$ is called the \emph{bond dimension}. To avoid confusion we will mark physical indices by Roman letters and bond indices, such as the matrix elements $(A^{(j)})^{i,\alpha}_\beta$, by Greek letters in this subsection. \\

Note that \eqref{MPS_def} implies that $\ket{\psi}$ has $dnD^2$ parameters instead of $d^n$. In principle $D$ can be chosen as large as it needs to be, allowing MPS to represent any state. But the advantage of this representation becomes clear when we consider states that follow area laws. For instance any state that has its Schmidt-rank bounded by $c$ for any bipartition of the system can be exactly expressed by an MPS with $D=\mathcal{O}(c)$. More generally for many relevant states it is sufficient to have an area law on von Neumann entropy to guarantee arbitrarily good approximation with $D = \mathcal{O}(\mathrm{poly}(n))$ \cite{Verstraete2006,Perez-Garcia2007,Hastings2007,Schuch2008}.

At this point it is useful to introduce the graphical notation for tensor network states. We will represent a tensor by a box with legs, with the number of legs corresponding to the order of the tensor. Each leg will correspond to a copy of $\mathcal{H}_S$ (leg pointing up), $\mathcal{H}_S^*$ (leg pointing down), $\mathcal{H}_B$ (leg pointing left) and $\mathcal{H}_B^*$ (leg pointing right). Here are a few examples to see what type of object corresponds to which diagram:

\begin{gather*}
\includegraphics[width=0.06\textwidth, valign=c, raise=3pt]{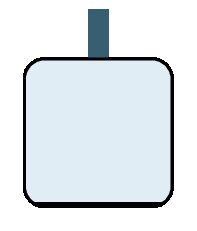} \in \mathcal{H}_S\;,\quad 
\includegraphics[width=0.06\textwidth, valign=c, raise=-3pt]{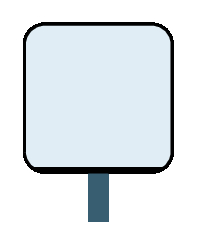} \in  \mathcal{H}_S^*\;,\\
\includegraphics[height=0.06\textwidth, valign=c]{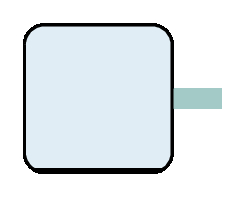} \in \mathcal{H}_B\;,\quad 
\includegraphics[height=0.06\textwidth, valign=c]{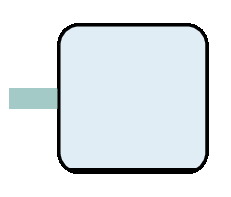} \in \mathcal{H}_B^*~, \\
\includegraphics[width=0.06\textwidth, valign=c]{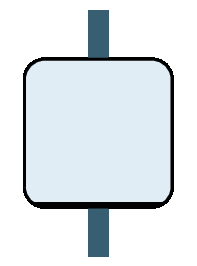} \in L(\mathcal{H}_S)\;,\quad 
\includegraphics[height=0.06\textwidth, valign=c]{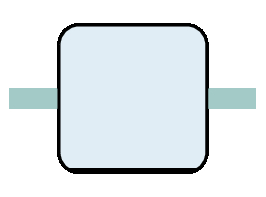} \in L(\mathcal{H}_B)~.
\end{gather*}

In this notation index contraction and partial tracing become:
\[
A^\alpha_\gamma B^\gamma_\beta = \includegraphics[height=0.06\textwidth, valign=c]{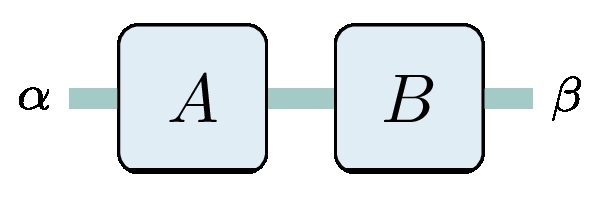}\;,\\ \Tr_{\mathcal{H}_{B_1}}[T] = \includegraphics[width=0.103\textwidth, valign=c, raise=2pt]{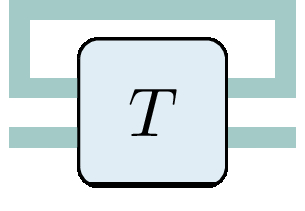}\;,
\]
where $A$ and $B$ are $1$-$1$ tensors, $T$ is a $2$-$2$ tensor and $\mathcal{H}_{B_1}$ refers to the first copy of $\mathcal{H}_{B}$ upon which $T$ acts. Using this notation, the MPS representation becomes (in the case of periodic boundary conditions):
\begin{equation}\label{MPS_diagram}
	\includegraphics[height=0.08\textwidth, valign=c]{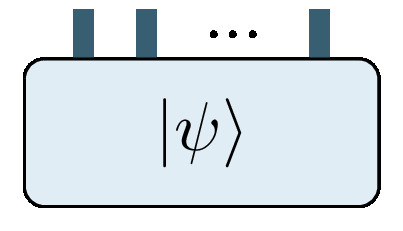}~ = \quad\includegraphics[height=0.08\textwidth, valign=c, raise=1pt]{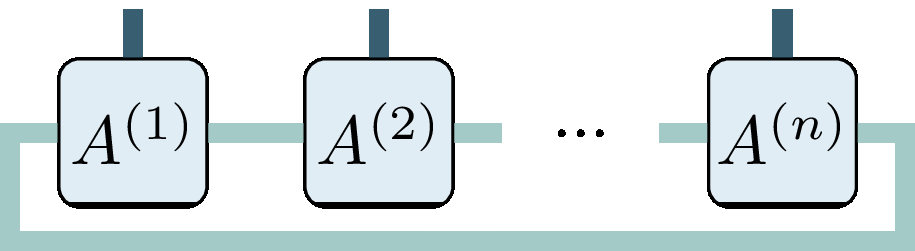}~.\nonumber
\end{equation}
Similarly, for $\bra{\psi}$ the result is:
\begin{equation}
	\includegraphics[height=0.08\textwidth, valign=c]{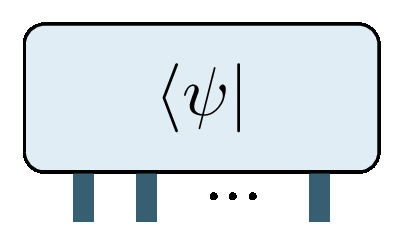}~ = \quad\includegraphics[height=0.08\textwidth, valign=c, raise=-1pt]{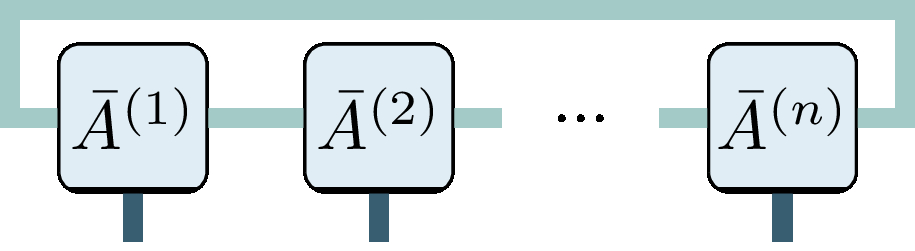}~.\nonumber
\end{equation}
The normalized density matrix corresponding to $\ket{\psi}$ is given by $\rho = \frac{1}{\braket{\psi}{\psi}} \ket{\psi}\bra{\psi}$. In the diagrammatic notation we get
\begin{equation}\label{rho_mps}
	 \hat\rho~ =\quad \includegraphics[width=0.35\textwidth, valign=c]{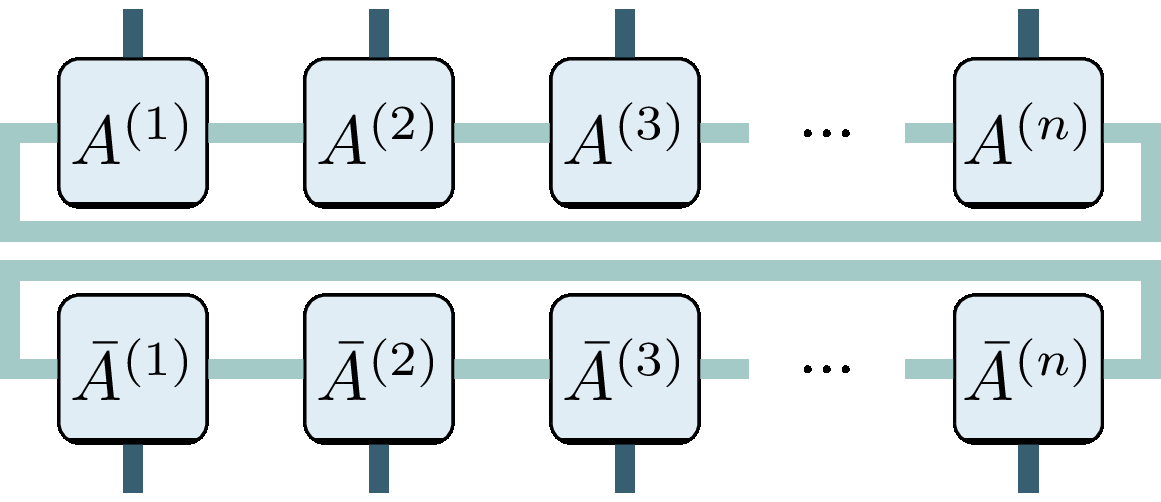}~,
\end{equation}
where we neglected the normalisation. We will always neglect the normalisation in the diagrams. In \eqref{rho_mps} it is easy to see that when we take a partial trace, we can consider the traced out site as an operator acting on $\mathcal{H}_B^{\otimes 2}$. This operator plays an important role and is called the \textit{transfer operator}. At a given site $j$ it is defined as:
\begin{equation}\label{def_T_formal}
	(T^{(j)})^{\alpha\beta}_{\gamma\delta}~ = ~(A^{(j)})^{i,\alpha}_\gamma (\bar{A}^{(j)})^\beta_{i,\delta}~ =~ \includegraphics[height=0.135\textwidth, valign=c]{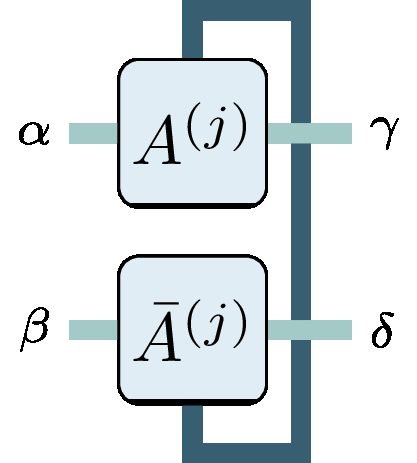}~,\nonumber
\end{equation}
or equivalently $T^{(j)} = (A^{(j)})^i\otimes(\bar A^{(j)})_i$, assuming Einstein summation convention. We can express the norm of $\ket{\psi}$ in terms of the transfer operator as
\begin{equation}\label{norm_psi}
	\begin{split}
		\braket{\psi}{\psi} 
		& = \Tr[T^{(1)}T^{(2)} \cdots T^{(n)}].
	\end{split}
\end{equation}
We emphasize here already that the transfer-operator is in general not a normal operator, i.e., it is not diagonalisable in an orthonormal basis, which can lead to technical difficulties.
However, the transfer-operator $T^{(j)}$ is in direct correspondence to the completely positive map
\begin{align}
	\mc T^{(j)}[X] := (A^{(j)})^i X (A^{(j)})^\dagger_i
\end{align}
acting on matrices $X\in \mathbb C^{D\times D}$. 
This correspondence, which we will refer to as \emph{Liouville representation} will be extremely useful later on. 
It arises by mapping a vector in $\ket{\phi}\ket{\psi}\in \mathcal H_B\otimes \mathcal H_B$ to the operator $\ket{\phi}\overline{\bra{\psi}}$, where the overline denotes complex conjugation in the basis $\{\ket{\alpha}\}$.
In particular, the spectra of $\mc T^{(j)}$ and $T^{(j)}$ are identical. From the theory of completely positive maps, it then follows that $\mc T^{(j)}$ (and $T^{(j)}$) has a positive (hence real) eigenvalue that is larger or equal to the modulus of all other eigenvalues. 
\begin{definition}[Gapped transfer operator]
	Let $T$ be a transfer-operator with eigenvalues (including multiplicities) $t_1\geq |t_2|\geq \cdots$. We say that $T$ is \emph{gapped} if $|t_j|<t_1$ for all $j\geq 2$ and define its gap as 
	\begin{align}
		\Delta := \frac{t_1 - |t_2|}{t_1} > 0.	
	\end{align}
\end{definition}
The property of being gapped is generic: If we choose the tensor $A^{(j)}$ randomly, then we obtain a gapped transfer operator with unit probability.
We will call the eigenvector to the eigenvalue $t_1$ the \emph{highest eigenvector} of $T$. 
Finally, we call a matrix product state \emph{translational invariant}, if all the tensors $A^{(j)}$ are identical, i.e., $A^{(j)}=A$ and $T^{(j)}=T$ for all $T$.

Any translationally invariant MPS with periodic boundary conditions is invariant under the transformation $A^i\mapsto X A^i X^{-1}$ for arbitrary invertible matrices $X$.
This \emph{gauge freedom} can be used to bring the tensors $A$ into a \emph{canonical form}, which is described in the following proposition. For more details in regard to this, see Ref.~\cite{Perez-Garcia2007}
(and Ref.~\cite{Fannes1992} for the equivalent results in the framework of finitely correlated states).
\begin{proposition}[Canonical form of translational invariant MPS, Theorem 4 of \cite{Perez-Garcia2007}]\label{can_form}
	Given a translational invariant MPS with periodic boundary conditions we can always decompose the matrices $A^i \in L(\mathcal{H}_B)$ as
	\begin{equation}\label{canonical_A}
		A^i = 
		\begin{pmatrix}
			\lambda_1 A^i_{(1)} & 0 & 0\\
			0 & \lambda_2 A^i_{(2)} & 0\\ 
			0 & 0 & \ddots
		\end{pmatrix}~,
	\end{equation}
	where $\lambda_j \geq 0$ for every $j$. We will call $\mathcal{H}_j$ the subspace corresponding to the $j$-th block, with $\mathcal{H}_i \perp \mathcal{H}_j$ for $i\neq j$.
	The $A^i_{(j)}$ blocks each satisfy the following conditions:
	\begin{itemize}
	 	\item[\textit{a.}] $P_j$ is the only fixed point of the irreducible, primitive, CP map $\mc T_j[X] = \sum_i A^i_{(j)}XA_{i,(j)}^{\dagger}$ (no sum on $j$), with $P_j$ the orthogonal projector onto $\mathcal{H}_j$.
	 	\item[\textit{b.}] $\mc T_j^*[\Lambda_j] = \Lambda_j$, where $\mc T_j^*$ denotes the adjoint map of $\mc T_j$. $\Lambda_j$ is such that $\left.\Lambda_j\right|_{\mathcal{H}_j}$ is positive and full rank and $\left.\Lambda_j\right|_{\mathcal{H}_j^c} = 0$.
	\end{itemize}
\end{proposition}
We will take the convention that $\lambda_1 \geq \lambda_2 \geq \lambda_3 \geq ...$ . 
Furthermore, normalization freedom allows us to choose $\lambda_1 = 1 > \lambda_2$, where the inequality follows if the transfer operator $T$ is gapped. 
Also observe that the spectrum of the transfer operator $T$ remains invariant under a gauge-transformation. 
Furthermore, whether the highest eigenvector of $T$ is a product state or not is also a gauge-invariant property. 
We will therefore from now on assume that the MPS is in canonical form and the transfer operator is normalized.

The canonical form does not imply that the CP-map $\mc T$ is block-diagonal with respect to the above block-decomposition.
Nevertheless, we have
\begin{align}
	\mc T[P_1] &= \sum_i \sum_{j,k}\lambda_j\lambda_k A^{(i)}_j P_1 {A^{(i)}_k}^\dagger = \sum_i \lambda_1^2 A^{(i)}_1 P_1 {A^{(i)}_1}^\dagger \nonumber \\
	&= \lambda_1^2 \mc T_1[P_1] = P_1.
\end{align}
So $P_1$ is a fixed-point of $\mc T$. If $T$ is gapped, then this fixed-point is unique. In particular, this implies that the normalized highest eigenvector of $T$ is given by
\begin{align}\label{eq:t1}
	\ket{t_1} = \frac{1}{\sqrt{d_1}} \sum_{j=1}^{d_1} \ket{j}\otimes\ket{j},
\end{align}
where $\{\ket{j}\}_{j=1}^{d_1}$ is an orthonormal basis of $\mc H_1$ and $d_1 = \dim(\mc H_1)$. This result tells us that $\ket{t_1}$ is a product state if and only if $d_1=1$ if and only if $P_1$ has rank one.
It is easy to see that this in turn implies that the MPS converges to a product state in the thermodynamic limit (in a weak sense). 
Similarly, we find that $\Lambda_1$ is the unique fixed point of the dual map $\mc T^*$. Later, we will often be interested in the asymptotic case map $\lim_{m\rightarrow \infty} \mc T^m$. 
If $T$ is gapped, its form is very simple, as shown by the following Lemma, which is a reformulation of a standard result on completely positive maps (see, for example, Ref.~\cite{Wolf2012}). 
\begin{lemma}\label{lemma:powersprojection}
	Let $T$ be normalized, gapped and in canonical form. Then
	\begin{align}
		\lim_{m\rightarrow\infty}\mc T^m[X] = P_1\Tr[\Lambda_1 X] =: \mc P_{\Lambda_1}[X]. 
	\end{align}
\end{lemma}

In the canonical form, we see that in general $\mc T$ does not map the identity to itself, but we always have $\mc T[\1]\leq \1$. We call a CP-map that has this property \emph{sub-unital}.
An important property of sub-unital maps that we will need later is that they fulfill a Schwarz-inequality (see, for example, Refs.~\cite{Choi1974,Bhatia2009,Wolf2012} for more discussions of Schwarz inequalities for CP-maps):
\begin{lemma}\label{schwarz}
	Let $A : L(\mathcal{H}) \rightarrow L(\mathcal{H}')$ be a completely positive, sub-unital map (i.e., $A[\mathbbm{1}] \leq \mathbbm{1}$), with  $\mathcal{H}$ and $\mathcal{H}'$ two finite-dimensional Hilbert spaces. Then for all $X \in L(\mathcal{H})$
	\begin{equation}\label{schwarz_ineq}
		A[X]A[X^{\dagger}] \leq A[XX^{\dagger}].
	\end{equation}
\end{lemma}
\begin{proof}
	Since $A$ is completely positive, by the Stinespring representation we can choose an operator $V$ such that $A[X] = V(X\otimes\mathbbm{1})V^{\dagger}$: for $K_j$ the Kraus operators of $A$, choose $V = \sum_j K_j\otimes\bra{j}$.	Because $A$ is sub-unital, we have
	\begin{equation}
		VV^{\dagger} = A[\mathbbm{1}] \leq \mathbbm{1}.
	\end{equation}
	Thus the eigenvalues of $VV^{\dagger}$ are smaller than or equal to $1$. Since the eigenvalues of $VV^{\dagger}$ and $V^{\dagger}V$ are either $0$ or the singular values squared of $V$ we also have $V^{\dagger}V \leq \mathbbm{1}$, implying
	\begin{align}
	A[X]A[X^{\dagger}] &= V(X\otimes\mathbbm{1})V^{\dagger}V(X^{\dagger}\otimes\mathbbm{1})V^{\dagger} \\
		&\leq V(X\otimes\mathbbm{1})\mathbbm{1}(X^{\dagger}\otimes\mathbbm{1})V^{\dagger} \\
		&= V(XX^{\dagger}\otimes\mathbbm{1})V^{\dagger} = A[XX^{\dagger}]~.
	\end{align}
\end{proof}

\subsection{R\'enyi entropies of sub-systems}
Let us now discuss R\'enyi entropies and how to compute them in the framework of MPS.
R\'enyi entropies provide a family of entropy-measures of quantum states, depending on a real parameter $\alpha$: 
\begin{definition}[R\'enyi entropies]
	\label{RS}
	Given a density matrix $\rho$ its \emph{R\'enyi-$\alpha$ entropy} $S_\alpha(\rho)$ is defined to be
	\begin{align} 
		S_\alpha(\rho) &= \frac{1}{1-\alpha}\log(\Tr[\rho^\alpha])
	\end{align}
	for every $0 < \alpha < \infty$, $\alpha \neq 1$.  For the values $\alpha = 0,1,\infty$ we define the corresponding quantities by continuity as $S_0(\rho) = \log(\rank[\rho])$, $S_1(\rho) = - \Tr[\rho \log(\rho)]$ and $S_\infty = -\log(\|\rho\|)$. 
\end{definition}
R\'enyi entropies are important, since different R\'enyi entropies measure different properties of the density matrix. 
We can notice that $S_1$ is the von Neumann entropy, while $S_0$ and $S_\infty$ are, respectively, the max- and min-entropies, which are often encountered in quantum information theory. $\|\rho\|$ is equal to the greatest singular value of $\rho$. Since $\rho$ is a density matrix its eigenvalues are equal to the singular values. Therefore if $\rho$ has spectral decomposition $\lambda_1 \geq \lambda_2 \geq \cdots$, we have $S_\infty(\rho) = -\log(\lambda_1)$. 
For all allowed values of $\alpha$, Rényi entropies have the following properties:
\begin{itemize}
	\item $0 \leq S_\alpha(\rho) \leq \log(\rank[\rho])=S_0(\rho)$ and $S_\alpha(\rho)=0$ if and only if $\rho$ is pure.
	\item Additivity: $S_\alpha(\rho\otimes\sigma) = S_\alpha(\rho) + S_\alpha(\sigma)$.
	\item Monotonicity: if $ \alpha \leq \beta$ then $S_\alpha(\rho) \geq S_\beta(\rho)$. 
	\item For $\alpha \geq \beta > 1$ the following inequality holds \cite{Beck1990,Wilming2019}:
	\begin{equation}\label{norm_behaviour}
		S_\infty(\rho) \geq \frac{\alpha-1}{\alpha}S_\alpha(\rho) \geq \frac{\beta-1}{\beta}S_\beta(\rho).
	\end{equation}
\end{itemize}

Using \eqref{norm_behaviour} together with monotonicity, we can lower-bound all Rényi entropies with the particularly well-behaved $S_2$: for all $\alpha \geq 0$ and all states $\rho$, we have
\begin{equation}\label{one_s_to_rule_them_all}
	S_\alpha\left(\rho\right) \geq S_\infty\left(\rho\right) \geq \frac{1}{2} S_2\left( \rho\right).
\end{equation}
The Rényi-2 entropy of a state is $S_2(\rho) = -\log(\Tr[\rho^2])$. We will now discuss how to compute the R\'enyi-2 entropy of a sub-system for a matrix-product state. 
We consider the sub-system given by simply taking every $k$-th spin of the chain. Given a spin chain of length $n$ described by the density matrix $\rho^{(n)}$, we therefore define $\rho^{(n)}_k$ to be the reduced state describing every $k$-th spin (for $k \in \mathbb{N}$ and $2\leq k\leq n$). 
In other words, we trace out from the global state every spin of which the site index is not a multiple of $k$. 
\newpage
If we denote by $j$ the site indices we have
	$$\rho^{(n)}_k = \Tr_{j\notin k[1,\ldots n]}[\rho^{(n)}]~.$$
%In the following, we will often omit the super-script $n$. 
We can admit without much loss of generality that $n$ is a multiple of $k$ and do this in the following. In this case, we have:
\begin{widetext}
$$\rho_k = \includegraphics[width=0.98\textwidth, valign=c]{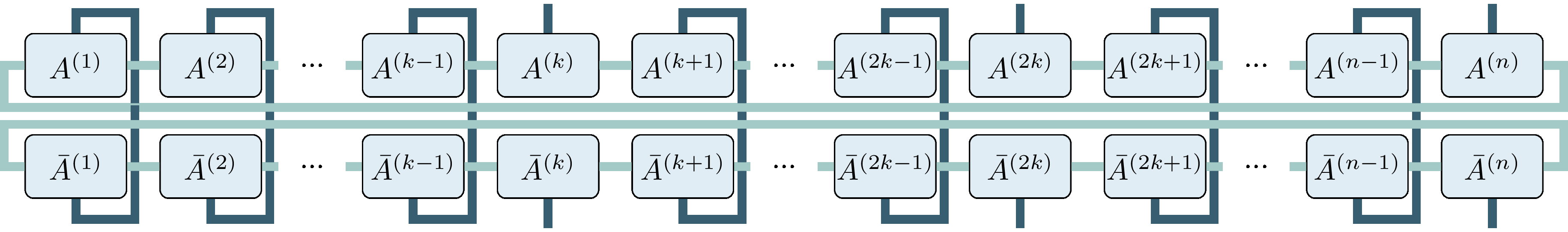}.$$
To make the structure more clear we can regroup the operators in the following way:
$$\includegraphics[width=0.72\textwidth, valign=c]{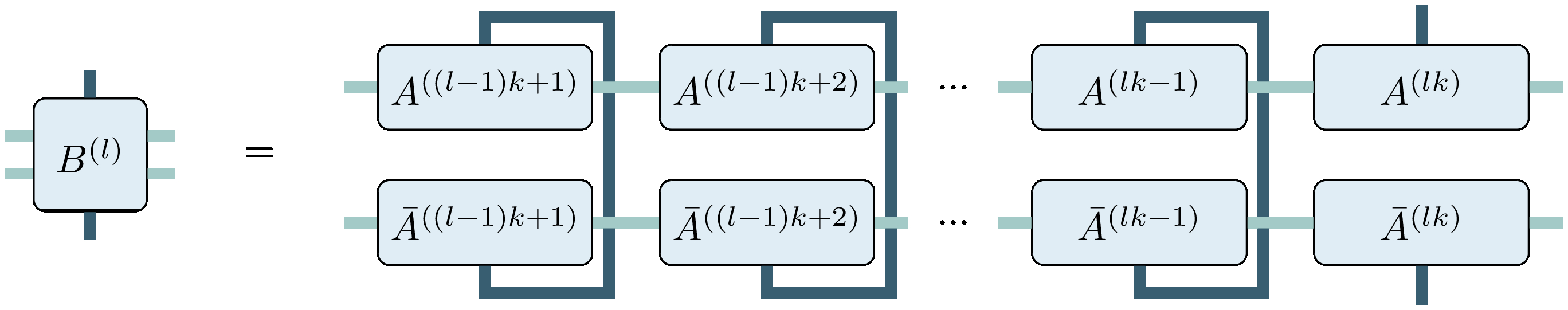}~.$$
\end{widetext}
This regrouping defines the operator $B^{(l)}:\mathcal H_S \otimes \mathcal H_B \otimes \mathcal H_B \rightarrow \mathcal H_S \otimes \mathcal H_B \otimes \mathcal H_B$.
We can now write $\rho_k^{(n)}$ as a product of the tensors $B^{(l)}$:
$$ \rho_k^{(n)} = \includegraphics[width=0.28\textwidth, valign=c]{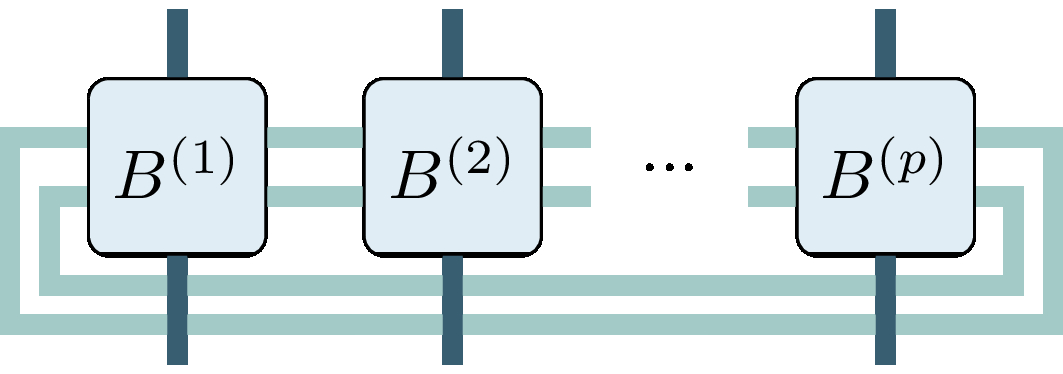}~, $$
where $n = kp$.  Note that on the r.h.s. we left out the normalization factor given by
\begin{align}
1/\Tr[T^{(1)}T^{(2)} \cdots T^{(n)}].
\end{align}
Now define an operator $\hat T^{(l)}_k$ by composing $B^{(l)}$ vertically with itself and taking the trace over the \emph{physical bonds}:
$$\includegraphics[width=0.23\textwidth, valign=c]{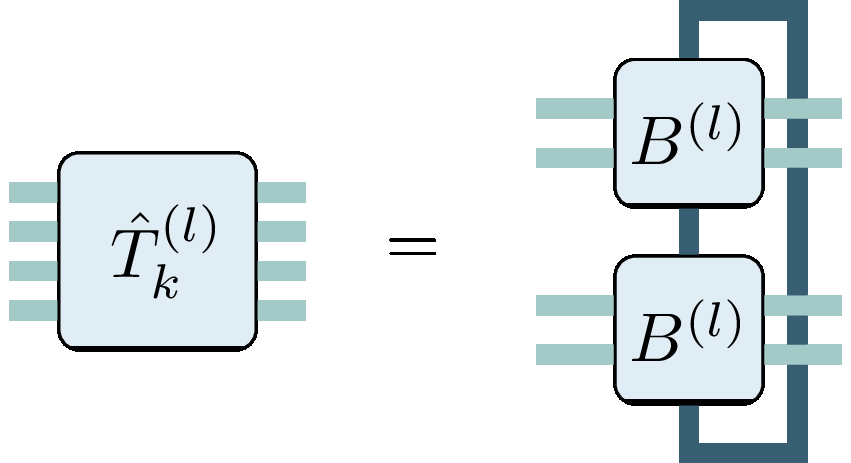}~.$$
For reasons that will become clear in the next section, we call $\hat T^{(l)}_k$ the \emph{twisted transfer operator}. We can use it to obtain a compact representation of $\Tr\big[(\rho_k^{(n)})^2\big]$:
\begin{align}\label{eq:fundamental}
	\Tr\Big[\Big(\rho_k^{(n)}\Big)^2\Big] = \frac{\Tr\big[\hat T^{(1)}_k \hat T^{(2)}_k \cdots \hat T^{(n/k)}_k \big]}{\Tr\big[T^{(1)}T^{(2)} \cdots T^{(n)}\big]^2},
\end{align}
which allows us to compute $S_2$ provided we have sufficient information about $\hat T^{(j)}_k$ and $T^{(j)}_k$. In fact we will see that the R\'enyi-2 entropy density only depends on the largest eigenvalue of the twisted transfer operator.

\subsection{The twisted transfer operators and higher R\'enyi entropies}
Form \eqref{eq:fundamental} it is clear that the twisted transfer operator plays a central role in estimating the R\'enyi entropies of $\rho_k$. 
We therefore now discuss the structure of this operator in more detail. 
We restrict to the translational invariant case, but it should be clear how the following statements generalize to the general case. In the following we therefore set $A^{(j)}=A,T^{(j)}=T,$ and $\hat T^{(l)}_k= \hat T_k$. 

It is useful to denote by $\mf{s}$ the \emph{swap} operator acting on $\mc H_B$:
\begin{align}
	\mf s \ket{\phi_1}\otimes \ket{\phi_2} = \ket{\phi_2}\otimes\ket{\phi_1},\quad \forall \ket{\phi_1},\ket{\phi_2}\in \mc H_B.\nonumber
\end{align}
The twisted transfer operator acts on $\mc H_B^{\otimes 4}$. We label the different copies of $\mc H_B$ from $1$ to $4$ and define the \emph{partial swap} operator
\begin{align}
	U := \1_{13}\otimes \mf s_{24},
\end{align}
where the subscripts indicate on which factors the operators act. Concretely, we have
\begin{align}
	U\ket{\phi_1}\ket{\phi_2}\ket{\phi_3}\ket{\phi_4} = \ket{\phi_1}\ket{\phi_4}\ket{\phi_3}\ket{\phi_2}.\nonumber
\end{align}
Note that $U= U^\dagger$ and $ U^2=\1$.
We can then write the twisted transfer operator $\hat T_k$ as
\begin{align}
 \hat T_k = (T\otimes T)^{k-1}  U (T\otimes T) U, 
\end{align}
which can be graphically represented as follows:
\begin{align}
	\includegraphics[width=0.37\textwidth, valign=c]{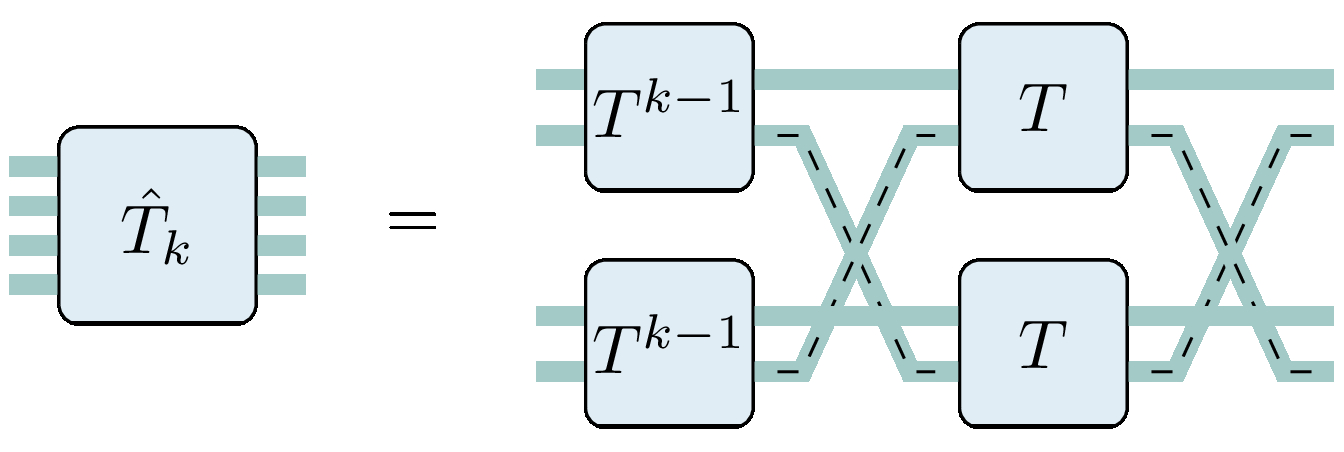}.\nonumber
\end{align}
A similar calculation as for the second R\'enyi entropy can be done for higher order R\'enyi entropies $S_l$ with $2\leq l\in \mathbbm{N}$, which are essentially computed by $\Tr\big[(\rho_k^{(n)})^l\big]$. 
One then finds that the result is expressed in terms of the twisted transfer operator of order $l$, which we define as
\begin{align}
	\tensor[^{}_l]{\hat T}{^{}_k}	:= (T^{\otimes l})^{k-1} U_l^{-1} T^{\otimes l} U_l, \quad U_l := \1_{135 \cdots} \otimes \mf s_{246\cdots},
\end{align}
where now $\mf s$ denotes the unitary \emph{shift operator} on $\mc H_B^{\otimes l}$, which acts by cyclically permuting the tensor factors:
\begin{align}
	\mf s \ket{\phi_1}\otimes\ket{\phi_2}\otimes\cdots\otimes\ket{\phi_l} = \ket{\phi_2} \otimes \ket{\phi_3}\otimes \cdots \otimes\ket{\phi_1}.\nonumber
\end{align}
Here, we decided to not indicate the $l$-dependence of $\mf s$ explicitly, as we will later mostly be concerned with the case $l=2$ and in any case $l$ will be obvious from the context. 
We clearly have $\hat T_k \equiv \tensor[^{}_2]{\hat T}{^{}_k}$ and $U\equiv U_2$.
Using the twisted transfer operator of order $l$, we then have
\begin{align}\label{eq:trrhol}
	\Tr[(\rho^{(n)}_k)^l] = \frac{\Tr\big[(\,\tensor[^{}_l]{\hat T}{^{}_k})^{n/k}\big]	}{\Tr[T^n]^l}. 
\end{align}

Later we will often work in the Liouville representation. Then the unitary $U_l$ is associated to right-multiplication by $\mf s$:
\begin{align}
	\mf U_l[X] = X\mf s 
\end{align}
and the twisted transfer-operator of order $l$ acts as
\begin{align}
	\tensor[^{}_l]{\mc T}{^{}_k}[X] = (\mc T^{\otimes l})^{k-1}\big[ \mc T^{\otimes l}[X\mf s]\mf s^\dagger \big]. 
\end{align}
To conclude this section, we compute the twisted transfer operator in the limits $k\rightarrow \infty$ and $m\rightarrow \infty$. 
\begin{lemma}\label{lemma:twistedTpowers}
Let $T$ be gapped and in canonical form. Then the twisted transfer-operator fulfills	
\begin{align}
	\lim_{m\rightarrow \infty} \lim_{k\rightarrow \infty}	\tensor[^{}_l]{\mc T}{^{m}_k}[P_1^{\otimes l}]	&= P_1^{\otimes l} \Tr\big[\Lambda_1^l\big]^2. 
\end{align}
	In particular, $P_1^{\otimes l}$ is the only eigenvector, with eigenvalue $\Tr[\Lambda_1^l]^2 = \exp(- 2(l-1)S_l(\Lambda_1))$.
	\begin{proof}
		Since $T$ is gapped, $\mc T^m$ converges to $\mc P_{\Lambda_1}$ in the limit $m\rightarrow \infty$ by Lemma~\ref{lemma:powersprojection}. Thus
	\begin{align}
		\lim_{m\rightarrow \infty} \lim_{k\rightarrow \infty}	\tensor[^{}_l]{\mc T}{^{m}_k}[P_1^{\otimes l}]  &= \lim_{m\rightarrow \infty} \mc P_{\Lambda_1}^{\otimes l}\big[ {\mc T^m}^{\otimes l}[P_1^{\otimes l} \mf s]\mf s^\dagger \big] \nonumber \\
		&= \mc P_{\Lambda_1}^{\otimes l}\Big[\mf s \mc P_{\Lambda_1}^{\otimes l}[\mf s^\dagger]\Big] \nonumber \\
		&= P_{1}^{\otimes l}\Tr\big[\Lambda_1^{\otimes l} \mf s\big]\Tr\big[\Lambda_1^{\otimes l} \mf s^\dagger \big]\nonumber \\
		&= P_1^{\otimes l}\Tr[\Lambda_1^l]^2,
	\end{align}
		where we used $\Tr[A_1\otimes A_2 \otimes\cdots\otimes A_l\mf s] = \Tr[A_1\cdots A_l]$ and $\mc P_{\Lambda_1}[X]=P_1 \Tr[\Lambda_1 X]$. 
	\end{proof}
\end{lemma}

\subsection{The shift-operator and completely positive maps}
It is apparent from the discussion of the twisted transfer operator, that the shift-operator in conjunction with tensor powers of completely positive maps are central to the behaviour of R\'enyi entropies. 
We now make some simple, but important observation about the relationship between completely positive maps and the shift-operator $\mf s$.
\begin{lemma}\label{lemma:symmetry}
	$\mf s$ is a symmetry of $\mc T^{\otimes l}$:
	\begin{align}
		\mc T^{\otimes l}[\mf s X\mf s^\dagger] = \mf s \mc T^{\otimes l}[X] \mf s^\dagger.
	\end{align}		
	Therefore, 
	\begin{align}
	[\mf s,\mc T^{\otimes l}[\mf s]]=[\mf s,\mc T^{\otimes l}[\mf s^\dagger]]=0.
	\end{align}
\begin{proof}
	That $\mf s$ is a symmetry can be seen from expanding $X$ into an operator basis in tensor-product form. For the commutation relations, take $X = \mf s$ or $X = \mf s^\dagger$ and use the symmetry property.
\end{proof}
\end{lemma}
A particularly useful implication of the previous Lemma is that $\mf s \mc T^{\otimes 2}[\mf s]$ is Hermitian, since in this case we have $\mf s= \mf s^\dagger$. 

\begin{lemma}\label{lemma:sEsnorm} Let $\mc E: L(\mc H) \rightarrow L(\mc H')$ be a sub-unital, completely positive map and $U\in L(\mc H), V\in L(\mc H')$ be unitaries. Then $\norm{V \mc E[U]}\leq 1$.
			In particular $\norm{\mc T^{\otimes l}[\mf s]\mf s^\dagger} \leq 1$. 
		\begin{proof} By unitary invariance of the operator norm we have
			\[\norm{V\mc E[U]}^2 = \norm{\mc E[U]}^2 = \max_\rho \Tr[\rho \mc E[U]^\dagger \mc E[U]],\]
			where the optimization is done over density matrices. We now use the Schwarz-inequality from Lemma~\ref{schwarz} to get
			$\mc E[U]^\dagger \mc E[U]\leq \mc E[U^\dagger U] \leq \1$. 
		\end{proof}

		\end{lemma}
The result of the previous Lemma together with the fact that $\mc T^{\otimes 2}[\mf s]\mf s$ is Hermitian has the following important implication:
\begin{lemma}\label{lemma:1D}
	Let $\mc T: L(\mc H)\rightarrow L(\mc H)$ be a unital CP-map with $\dim(\mc H)=D$ and $0<\rho \in L(\mc H\otimes \mc H)$ be a density matrix with full rank. Suppose that $\1$ is the unique fixed-point of $\mc T$. Then
	\begin{align}
	\Tr[\rho \mc T\otimes \mc T[\mf s]\mf s] = 1 \quad\Rightarrow\quad D=1.
	\end{align}
	\begin{proof}
		Let $\rho = \sum_i p_i \proj{i}$ be the spectral decomposition of $\rho$ and $r_i := \bra{i}\mc T[\mf s]\mf s \ket{i}$. From Lemma~\ref{lemma:sEsnorm} we know that $-1\leq r_i\leq 1$. 
		Then $\sum_i p_i r_i =1$ implies $r_i=1$ for all $i$. But then $\mc T\otimes \mc T[\mf s]\mf s = \1$ and $\mc T\otimes \mc T[\mf s]=\mf s$. Since $\1$ is the unique fixed-point of $\mc T\otimes \mc T$ by assumption, this implies $\mf s = \pm\1$. This is only possible if $D=1$, since for $D>1$ there exist both symmetric and anti-symmetric states in $\mc H\otimes \mc H$ and therefore $\mf s$ has eigenvalues $+1$ and $-1$. 
	\end{proof}
\end{lemma}
Finally, the following Lemma provides a simple relation between the singular values of a sub-unital completely positive map. 
\begin{lemma}\label{lemma:swap-singular-values} Let $\mc T:L(\mc H)\rightarrow L(\mc H)$ be a sub-unital, completely positive map with matrix-representation $T$ and $D=\dim(\mc H)$. Then 
			\[ \Tr[\mf s \mc T\otimes \mc T[\mf s]]=\sum_i s_i(\mc T)^2 \leq D^2.\]
		Furthermore, if $\1$ is the unique fixed-point of $\mc T$ and $\sum_i s_i^2(\mc T)=D^2$, then $D=1$.
		\begin{proof}	
			Pick a basis and write $\mf s = \sum_{i,j=1}^D |i\rangle\langle j|\otimes |j\rangle\langle i|$. Then
			\begin{align}
				\Tr[\mf s\mc T\otimes \mc T[\mf s]] &= \sum_{i,j} \Tr[\mf s \mc T[|i\rangle\langle j|]\otimes \mc T[|j\rangle\langle i|]] \\
				&=\sum_{i,j} \Tr[\mc T[|i\rangle\langle j|] \mc T[|j\rangle\langle i|]] \\
				&=\sum_{i,j} \Tr[ (|j\rangle\langle i|)^\dagger \mc T^*\mc T[|j\rangle\langle i|]] \\
				&= \Tr[T^* T] = \sum_i s_i(T)^2 = \sum_i s_i(\mc T)^2. 
			\end{align}
			where we used that the operators $|i\rangle\langle j|$ provide an orthonormal basis with respect to the Hilbert-Schmidt inner product and that the trace of an operator coincides with the sum of its eigenvalues. 
			On the other hand, we know that $\norm{\mf s \mc T\otimes \mc T[\mf s]}\leq 1$ by Lemma~\ref{lemma:sEsnorm}. 
			Hence $\Tr[\mf s \mc T\otimes \mc T[\mf s]]\leq \Tr[\1] \norm{\mf s \mc T\otimes \mc T[\mf s]} \leq D^2$.
			Finally, if $\1$ is the unique fixed-point of $\mc T$, then we can use Lemma~\ref{lemma:1D} with $\rho = \1/D^2$ to conclude that $\sum_i s_i(\mc T)^2 = D^2$ implies $D=1$.			
		\end{proof}
\end{lemma}

\section{Extensivity of R\'enyi entropies}
\label{sec:translational-invariant}
We now give a precise statement of our first main result and discuss its proof. 
As before, we call the eigenvector $\ket{t_1}$ corresponding to the largest eigenvalue $t_1$ of $T$ the \emph{highest eigenvector}.
Furthermore, as above, we denote by $\mc T$ the CP-map corresponding to the transfer operator $T$.
We first express the R\'enyi-$l$ entropy densities in terms of the largest eigenvalue of the twisted transfer operator of order $l$:
%%%
\begin{proposition}\label{proposition:renyil}
Let $\hat\rho_k^{(n)}$ as above for a translational invariant MPS with gapped transfer operator $T$ and periodic boundary conditions. Then
	 \begin{equation}\label{eq:liminf}
		 \liminf_{n\rightarrow \infty} \frac{k}{n}S_l\Big(\rho_k^{(n)}\Big) = -\frac{1}{l-1}\log\big(\big|\tensor[^{}_l]{\hat{t}}{}^{(k)}_1\big|\big)
	 \end{equation}
	 where $\tensor[^{}_l]{\hat t}{}^{(k)}_1$ is the greatest eigenvalue of the twisted transfer operator $\tensor[^{}_l]{\hat T}{^{}_k}$. 
	 \begin{proof} See section~\ref{sec:proofliminf}.
	\end{proof}
\end{proposition}
The proposition tells us that the R\'enyi entropy densities are entirely determined by the largest eigenvalue of the twisted transfer operators. 
A particularly simple case is given if we also take the limits $k\rightarrow \infty$ and the block size $m\rightarrow \infty$. From Lemma~\ref{lemma:twistedTpowers}, we then find
\begin{corollary}
In the limit $k\rightarrow \infty$ and $m\rightarrow \infty$, the entropy per block of $m$ spins fulfills
	 \begin{equation}\label{eq:liminfasymptotic}
		 \lim_{m\rightarrow \infty}	 \lim_{k\rightarrow \infty}	 \liminf_{n\rightarrow \infty} \frac{k}{n}S_l\Big(\rho_k^{(n)}\Big) = 2 S_l(\Lambda_1).
	 \end{equation}
\end{corollary}
The above calculation can be seen as a generalization of Theorem 6 in Ref~\cite{Perez-Garcia2007}, where it was shown that the reduced state of half the total system has the same spectrum as $\Lambda_1\otimes\Lambda_1$ in the thermodynamic limit. 
The corollary shows that in the limit $m,k\rightarrow\infty$ the (liminf of the) entropy density vanishes if and only if $\Lambda_1$ is a pure state, which is true if and only if $\dim(\mc H_1)=1$, which implies that state describing the blocks of $m$ spins approaches a product state in the thermodynamic limit. 
The following theorem, which is one of our main results, essentially reaches the same conclusion without taking the limits $k\rightarrow \infty$ and $m\rightarrow \infty$. 
\begin{theorem}\label{thm:eq_rel_null}
	Let $\hat\rho_k^{(n)}$ as above for a translational invariant MPS with gapped transfer operator $T$ and periodic boundary conditions. Then
	 \begin{equation}\label{s2_liminf}
		 \liminf_{n\rightarrow \infty} \frac{k}{n}S_2\Big( \rho_k^{(n)}\Big) = -\log\big(\big|\hat{t}^{(k)}_1\big|\big)
	 \end{equation}
	 where $\hat t_1^{(k)}$ is the greatest eigenvalue of the twisted transfer operator $\hat T_k$. Furthermore the following are equivalent:
	 \begin{enumerate}[i)]
		 \item\label{thm1:1} The highest eigenvector of $T$ is a product state. Equivalently, the eigenoperator of $\mc T$ with eigenvalue $t_1$ has rank 1.
		 \item\label{thm1:2} For all $k_0 \in \mathbb{N}$ there exists $k\geq k_0$ such that $\liminf\limits_{n\rightarrow \infty} \frac{k}{n}S_2(\rho_k^{(n)}) = 0$.
         	 \item\label{thm1:3} The Rényi-$\alpha$ entropy of $ \rho_k^{(n)}$ approaches $0$ in the thermodynamic limit for $\alpha > 1$ for all $k\geq 2$.
       		\item\label{thm1:4} The reduced density matrix of any fixed finite subsystem approaches a (site-wise) product state in the thermodynamic limit.
	 \end{enumerate}
\end{theorem}
%%%
The consequences of this theorem are quite clear. We have two distinct situations for translation invariant MPS states with gapped transfer operator:
\begin{enumerate}
	\item The highest eigenvector of the transfer operator is a product state and the total state weakly approaches a product state of each site, in the sense that for any local observable $A$ 
		\begin{align}
			\lim_{n\rightarrow\infty} \bra{\psi}A\ket{\psi} = \lim_{n\rightarrow\infty} \Tr[\ket{\phi}\bra{\phi}^{\otimes n} A]
		\end{align}
		for some $\ket{\phi}$.
	\item The highest eigenvector of $T$ is not a product state and therefore there must exist a finite value $k_0$ such that 
		\begin{align}
			\liminf\limits_{n\rightarrow \infty} \frac{k}{n}S_2(\rho_k^{(n)}) > 0
		\end{align}
		for all $k\geq k_0$, and by \eqref{one_s_to_rule_them_all} all Rényi entropies of $\rho_k^{(n)}$ become extensive. 
\end{enumerate}
It should be clear that the second situation is the generic one. That is, if we choose the entries of the MPS-tensor $A$ at random,  then with unit probability the entropy density is strictly positive.

For readers aware of the formalism of MPS or finitely correlated states, the implication \ref{thm1:1} $\Rightarrow$ \ref{thm1:4} should not be surprising: Statement \ref{thm1:1} says that the state in the thermodynamic limit can be obtained using a MPS with bond-dimension $D=1$. But states with bond-dimension $D=1$ are product states. 
The truly interesting statement (and the one difficult to prove) therefore is statement \ref{thm1:2} a vanishing R\'enyi entropy-density for $\alpha>1$ implies that this is the case. 
A priori, a vanishing entropy-density could also be achieved with R\'enyi entropies $S_\alpha$ growing as $n^\kappa$ with $\kappa<1$. 
Interestingly, the theorem implies that if the transfer operator is gapped, this situation cannot occur: either the entropies converge to zero or they grow extensively.

Before coming to the proof of the Theorem, let us show a simple example illustrating the necessity of the assumption that the transfer operator is gapped.  Let us take the  MPS defined by
\begin{equation}\label{ex_A}
	A = \ket{0}_S\ket{0}\bra{0}_B + \sqrt{\beta}\ket{1}_S\ket{1}\bra{1}_B.
\end{equation}
for some $1\geq \beta \geq 0$. Therefore the transfer operator is
\begin{equation}
	T = \ket{00}\bra{00} + \alpha \ket{11}\bra{11}.
\end{equation}
The twisted transfer operator is given by
\begin{equation}
	\hat T_2 = \ket{0}\bra{0}^{\otimes 4} + \beta^{2k} \ket{1}\bra{1}^{\otimes 4}.
\end{equation}
Since $A^1A^0 = A^0A^1 = 0$, the state corresponding to \eqref{ex_A} is
\begin{equation}
	\ket{\Psi} = \ket{00...0} + \beta^{n/2}\ket{11...1},
\end{equation}
with $n$ the number of spins. 
For $\beta < 1$ the highest eigenvector of the transfer operator is $\ket{00}$. 
As Theorem~\ref{thm:eq_rel_null} predicts, for any local observable $A$ it is clear that $\lim_{n\rightarrow\infty}\Tr[\ket{\psi}\bra{\psi}A] = \lim_{n\rightarrow\infty}\Tr[\ket{0}\bra{0}^{\otimes n}A]$. At the same time, we have $\hat t_1=1$ and the entropy of any $\hat \rho_k$ approaches zero in the thermodynamic limit. 
On the other hand, if $\beta = 1$ we obtain the GHZ state and the transfer operator is not gapped any more. 
There is no unique highest eigenvector of $T$. 
It could be chosen to be $\ket{00}$ which is a product state, however any reduced density matrix of the GHZ state has all R\'enyi entropies equal to $\log(2)$. 
Therefore the predictions of Theorem~\ref{thm:eq_rel_null} break down. Thus, the gap assumption is necessary for the theorem.

\subsection{Proof of Proposition~\ref{proposition:renyil}}
\label{sec:proofliminf}
We start by stating the two following Lemmas, whose proof can be found in Appendix~\ref{app:lemmas}.
\begin{lemma}
	\label{trace_ineq}
	Let $A$ be an endomorphism on $\mathcal H \cong \mathbb C^d$ such that its greatest eigenvalue fulfills $|a_1| \leq 1$. 
	Then $\lim_{n\rightarrow\infty}\log|\Tr[A^n]| = 0$ if and only if $A$ has greatest eigenvalue $|a_1| = 1$ and $|a_2| < 1$,
	where $a_2$ is the second largest eigenvalue of $A$.
\end{lemma}
\begin{lemma}
	\label{lim_sup}
	Let $\{\alpha_i\}_{i=1}^n$ be a collection of complex numbers such that $r = \max\limits_{i=1,...,n} |\alpha_i| >0$. 
	Then $$\limsup_{k\rightarrow\infty} \left|\sum_{i=1}^{n} \alpha_i^k \right|^{1/k} = r.$$
\end{lemma}
With these Lemmas at hand, we can now show~\eqref{eq:liminf}.
By \eqref{eq:trrhol} we have 
\begin{align}
	\Tr[(\rho_k^{(n)})^l] = \frac{\Tr[(\, \tensor[^{}_l]{\hat T}{^{}_k})^{n/k}]}{\Tr[T^n]^l}.
\end{align}
Since $T$ is gapped and we normalize, we can assume without loss of generality that $t_1 = 1$ and $|t_2| < 1$. 
Therefore, by Lemma~\ref{trace_ineq}, $\lim_{n\rightarrow\infty} \log\left|\Tr[T^n]\right| = 0$, which implies
	\begin{align} 
		\lim_{n\rightarrow\infty}S_l(\rho_k^{(n)}) &= \lim_{n\rightarrow\infty}-\frac{1}{l-1}\log\left(  \frac{\Tr[(\, \tensor[^{}_l]{\hat T}{^{}_k})^{n/k}]}{\Tr[T^n]^l}\right)\nonumber\\
		&= \lim_{n\rightarrow\infty}-\frac{1}{l-1}\log \Tr[(\, \tensor[^{}_l]{\hat T}{^{}_k})^{n/k}]~,\label{s2_lim}
	\end{align}
where we used the fact that $\Tr[T^n] > 0$.
We can now compute the entropy density in the thermodynamic limit. 
To simplify notation, we will simply write $\hat t_i := \tensor[^{}_l]{\hat t}{}^{(k)}_i$ for the eigenvalues of $\tensor[^{}_l]{\hat T}{^{}_k}$, sorted  with non-increasing absolute value. 
We then take the trace in an orthonormal basis such that $\hat t_i = \bra{i}\tensor[^{}_l]{\hat T}{^{}_k} \ket{i}$. This is always possible using the Schur-decomposition. We then obtain:
\begin{align}
\lim_{n\rightarrow \infty}\frac{k}{n} S_l\left(\rho_k^{(n)}\right)  &=
	\lim_{n\rightarrow \infty}\frac{-k}{n}\frac{1}{l-1}\log\left| \sum_{i=1}^{D^4} \hat{t}_i^{n/k} \right| \\
	&=	\lim_{n\rightarrow \infty} -k \frac{1}{l-1}\log\left| \sum_{i=1}^{D^4} \alpha_i^{n} \right|^{\frac{1}{n}}
\end{align}
with $\alpha_i = \hat{t}_i^{1/k}$. Then, by Lemma~\ref{lim_sup} we obtain
	\begin{align}\label{lim_inf_s2_formula}
		\liminf_{n\rightarrow \infty}\frac{k}{n} S_l(\rho_k^{(n)}) & = -k \frac{1}{l-1}\limsup_{n\rightarrow \infty} \log\left| \sum_{i=1}^{D^4} \alpha_i^{n} \right|^{\frac{1}{n}}\\
			& = -k\frac{1}{l-1}\log|\alpha_1| \\
			& = -\frac{1}{l-1}\log|\hat{t}_1|,
	\end{align}
showing \eqref{eq:liminf} and \eqref{s2_liminf}.
It is worth to noticing that this calculation proves $|\hat t_1| \leq 1 = t_1$, because otherwise we would have negative entropy densities.

\subsection{Proof of Theorem~\ref{thm:eq_rel_null}}
We will separate the proof of Theorem~\ref{thm:eq_rel_null} in several parts. 
The main difficulty of the proof consists of proving the equivalence of the statements \ref{thm1:1}--\ref{thm1:3},  which is done in section~\ref{sec:1to3}. 
Section~\ref{sec:1to4} then proves the equivalence of statement \ref{thm1:1} with statement \ref{thm1:4}. 
Throughout, remember that we consider $\hat T_k = \tensor[^{}_2]{\hat T}{}_k$. For simplicity of notation, we will label the eigenvalues of this operator with $\hat t_i$, such that $|\hat t_1|\geq |\hat t_2| \geq \cdots$.

\subsubsection{Equivalence of \ref{thm1:1} -- \ref{thm1:3}}
\label{sec:1to3}
We now prove the equivalence of statements \ref{thm1:1} -- \ref{thm1:3}. 
We first show that the largest eigenvalue of the twisted transfer operator $\hat T_k$ has at most one eigenvalue of magnitude $1$.
\begin{lemma}\label{one_eigenvalue}
	Assume that the transfer operator is gapped. Then the twisted transfer operator $\hat T_k$ has at most one eigenvalue of magnitude $1$, and if $|\hat t_1| = 1$ then $\hat t_1 = 1$.
\end{lemma}
\begin{proof}
	Considering that by definition
	\begin{equation}
		\Tr[(\hat \rho_k^{(n)})^2] = \frac{\Tr[\hat T_k^{n/k}]}{\Tr[T^n]^2}~,
	\end{equation}
	we can notice that for large values of $n$ the denominator on the RHS is exponentially close to $1$ (since $T$ is gapped and $t_1=1$) and always has a real positive value. 
	But then the numerator has to be a real and (strictly) positive number because $0 < \Tr[(\hat \rho_k^{(n)})^2] \leq 1$ for all values of $n$ and $k$. 
	By Lemma~\ref{unit_circle_powers} in the appendix, for any $n\in \mathbb{N}$ and any finite collection $\{z_i\}_i$ of numbers on the unit circle, and for all $\epsilon > 0$ one can find $l\in \mathbb{N}$ with $l\geq n$ such that $|(z_i)^l - 1| < \epsilon$ for all $i$.
	
	By taking $\{z_i\}_i = \{\hat t_i\; : \;|\hat t_i| = 1\}$, the condition $\Tr[(\hat \rho_k^{(n)})^2] \leq 1$  is infringed unless $\left|\{z_i\}_i\right| = 1$. 
	This allows us to conclude that $|\hat t_2| < 1$. 
	Combining this with the consideration that $\Tr[\hat T_k^{n/k}] > 0$ for all $n$ and $k$ we can conclude that if $|\hat t_1| = 1$ then it is unique and $\hat t_1=1$.
\end{proof}

We are now in position to state and prove following proposition, which allows us to prove the equivalence of statements~\ref{thm1:1}--\ref{thm1:3} of Theorem~\ref{thm:eq_rel_null}.
\begin{proposition}\label{eq_rel_tech}
	Let $\rho_k^{(n)}$ and $T$ be as in Theorem~\ref{thm:eq_rel_null}. Then the following are equivalent:
	\begin{enumerate}
		\item The highest eigenvector $\ket{t_1}$ of $T$ is a product state.
		\item For all $k_0 \in \mathbb{N}$ there exists $k\geq k_0$ such that $\liminf\limits_{n\rightarrow \infty} \frac{k}{n}\s_2(\hat \rho_k^{(n)}) = 0$,
		\item $\liminf\limits_{n\rightarrow \infty} \frac{k}{n}\s_2(\hat \rho_k^{(n)}) = 0$ for all $k\geq 2$,
		\item $\lim\limits_{n\rightarrow \infty} \s_2(\hat \rho_k^{(n)}) = 0$ for all $k\geq 2$.
	\end{enumerate}
\end{proposition}
Before we prove the proposition. Let us use it to prove the equivalence of statements~\ref{thm1:1}--\ref{thm1:3} of Theorem~\ref{thm:eq_rel_null}. 
The step \ref{thm1:3}$\Rightarrow$\ref{thm1:2} follows directly from the definitions. 
The step \ref{thm1:2}$\Rightarrow$\ref{thm1:1} is one of the statements of Proposition~\ref{eq_rel_tech}.
For the step \ref{thm1:1}$\Rightarrow$\ref{thm1:3}, first observe that Proposition~\ref{eq_rel_tech} implies that $\lim_{n\rightarrow \infty} S_2(\rho^{(n)}_k) =0$ for all $k\geq 2$.  By \eqref{norm_behaviour}, together with non-negativity and monotonicity of Rényi entropies, we have
\begin{equation}
	\s_2(\rho_k^{(n)}) \geq \s_\infty(\rho_k^{(n)}) \geq \frac{\alpha-1}{\alpha}\s_\alpha(\rho_k^{(n)}) \geq 0
\end{equation}
for all $\alpha > 1$. Since the l.h.s converges to $0$ in the thermodynamic limit, we obtain the desired result.

\begin{proof}[Proof of Proposition~\ref{eq_rel_tech}]
	($\textit{4.} \Rightarrow \textit{3.} \Rightarrow \textit{2.}$) These relations follow trivially from the definition of the statements.
	
	($\textit{1.} \Rightarrow \textit{4.}$) By assuming statement $(\textit{1.})$ we have $\ket{t_1} = \ket{\phi}\ket{\varphi}$, therefore 
	\begin{align}\label{t1_prod_evU}
		U\ket{t_1}_{12}\ket{t_1}_{34} &= \left(\mathbbm{1}_{13}\otimes\mf{s}_{24}\right)\ket{\phi}_1\ket{\varphi}_2\ket{\phi}_3\ket{\varphi}_4 \nonumber \\
		&=\ket{\phi}_1\ket{\varphi}_2\ket{\phi}_3\ket{\varphi}_4 = \ket{t_1}_{12}\ket{t_1}_{34},
	\end{align}
	where the indices indicate to which copy of $\mathcal{H}_B$ each vector belongs.
	But then $\hat T_k \ket{t_1}\ket{t_1} = \ket{t_1}\ket{t_1}$, meaning that $\hat T_k$ has $1$ as an eigenvalue and $\ket{t_1 t_1}$ is its corresponding eigenvector. 
	Since entropy is non-negative, by \eqref{s2_liminf} we have $|\hat t_1| \leq 1$, therefore we have $\hat t_1 = 1$. 
	By Lemma~\ref{one_eigenvalue}, $|\hat t_2| < 1$ for all $k\geq 2$, hence we can apply  Lemma~\ref{trace_ineq} to \eqref{s2_lim} and obtain the desired result.

($\textit{2.} \Rightarrow \textit{1.}$) This step is quite involved. Let us therefore first give a broad outline of the proof.
	We start by assuming that $|\hat t_1|=1$ and study the corresponding eigenvalue equation in Liouville space, in which $T$ is represented by the CP-map $\mc T$ ({\bfseries Step 1}).
	From this equation and using the Schwarz inequality, we first show that the highest eigenoperator $X$ of $\hat{\mc T}_k$ (which corresponds to the highest eigenvector of $\hat T_k$) has to be supported on $\mc H_1\otimes \mc H_1$, which allows us to reduce the analysis to the case of a single block, on which $\mc T$ is given by a primitive, unital CP-map ({\bfseries Step 2}). 
	Making use of the gap of $T$ and the asymptotic behaviour as $k\rightarrow \infty$, we can then show that $X=P_1\otimes P_1$ up to exponentially small corrections in $k$ ({\bfseries Step 3}). This in turn lets us show that $\mf{s}|_{\mc H_1\otimes \mc H_1}=P_1\otimes P_1\geq 0$, which is only possible if $\mc H_1$ is one-dimensional ({\bfseries Step 4}). Finally, we conclude that the highest eigenvector of $T$ is a product-state ({\bfseries Step 5}). Since $T$ is gapped we assume without loss of generality that $1 = \lambda_1 > \lambda_2 \geq \lambda_3 \geq \cdots$. \\ 

	{\bfseries Step 1.} As indicated already, the proof relies on the Liouville representation, which can be summed up as:
\begin{equation}\label{Liouville}
		\ket{\alpha \beta} \leftrightarrow \ket{\alpha}\overline{\bra{\beta}},
\end{equation}
where $\leftrightarrow$ signifies that this operation is a bijection and $\overline{\bra{\beta}} = \ket{\beta}^T$. Furthermore, this operation commutes with tensor products. Importantly, the Liouville representation preserves composition of linear maps. For the objects relevant for our discussion, we have:
\begin{equation}\label{new_ops}
		\begin{split}
			T = \sum_i A^i \otimes \bar A_i \qquad  &\leftrightarrow\qquad \mc T[\,\cdot\,] = \sum_i A^i\,\cdot\,A_i^{\dagger}~,\\
			U = \mathbbm{1}_{13}\otimes \mf{s}_{24}\qquad &\leftrightarrow\qquad \mc{U}[\,\cdot\,] = \,\cdot\,\,\mf{s}.
		\end{split}
	\end{equation}
	By \eqref{s2_liminf} statement $(\textit{2.})$ is equivalent to $|\hat t_1| = 1$. By Lemma~\ref{one_eigenvalue} we conclude that $\hat t_1 = 1$ and is unique.
	This allows us to reformulate the assumption as: ``for all $k_0\in \mathbb{N}$ there exists $k\geq k_0$ such that the equation 
	\begin{equation}\label{proof_eq_vec}
		(T\otimes T)^{k-1}U(T\otimes T) U\ket{x} = \ket{x} 
	\end{equation}
	has a unique solution up to a scalar factor''.
	
	By defining $\mc E = \mc T\otimes\mc T$, we can write	\eqref{proof_eq_vec} in the Liouville representation as
	\begin{equation}\label{proof_eq_op}
		\mathcal{E}^{k-1}\circ\mc U\circ\mathcal{E}\circ\mc U\left[X\right] 
		= \mathcal{E}^{k-1}\left[\mathcal{E}[X\mf{s}]\mf{s}\right] = X
	\end{equation}
	with $X$ a linear operator on $\mathcal{H}_B^{\otimes 2}$. 
	Statement $(\textit{2.})$ implies that for all $k_0\in \mathbb{N}$ there exists $k\geq k_0$ such that there exists a unique $X$ (up to a multiplicative factor) solving \eqref{proof_eq_op}. In principle $X$ depends on $k$, but we suppress this dependence in our notation. 

	{\bfseries Step 2.} We now make use of the canonical form in Proposition~\ref{can_form}. Since $\mc E = \mc T\otimes \mc T$ is sub-unital, it fulfills the Schwarz inequality of Lemma~\ref{schwarz}.
	For some doublet $\beta = (i,j)$ we now define $\Lambda_\beta = \Lambda_i \otimes \Lambda_j$, which is a full rank operator over $\mathcal{H}_\beta = \mathcal{H}_i\otimes\mathcal{H}_j$ and is a fixed point of the (Hilbert-Schmidt) adjoint $\mc E_\beta^*$ of $\mathcal{E}_\beta = \mc T_i\otimes\mc T_j$. 
	Using \eqref{proof_eq_op} and that CP-maps preserve hermiticity, we now observe
	\begin{align}
			\Tr\left[\Lambda_\beta XX^{\dagger}\right] & = \Tr\left[\Lambda_\beta \mathcal{E}^{k-1}\left[\mathcal{E}[X\mf{s}]\mf{s}\right] (\mathcal{E}^{k-1}\left[\mathcal{E}[X\mf{s}]\mf{s}\right])^{\dagger}\right] \nonumber \\
			& = \Tr\left[\Lambda_\beta \mathcal{E}^{k-1}\left[\mathcal{E}[X\mf{s}]\mf{s}\right] \mathcal{E}^{k-1}\left[(\mathcal{E}[X\mf{s}]\mf{s})^{\dagger}\right]\right].\nonumber
	\end{align}
	By repeated use of $\mf s^2=\1$ and the Schwarz inequality this implies 
	\begin{equation}\begin{split}
			\Tr\left[\Lambda_\beta XX^{\dagger}\right] & \leq \Tr\left[\Lambda_\beta \mathcal{E}^{k-1}\left[\mathcal{E}[X\mf{s}]\mf{s} (\mathcal{E}[X\mf{s}]\mf{s})^{\dagger}\right]\right] \\
			& = \Tr\left[\Lambda_\beta \mathcal{E}^{k-1}\left[\mathcal{E}[X\mf{s}] \mathcal{E}[(X\mf{s})^{\dagger}]\right]\right] \\
			& \leq \Tr\left[\Lambda_\beta \mathcal{E}^{k}\left[X\mf{s} (X\mf{s})^{\dagger}\right]\right] \\
			& = \Tr\left[\mathcal{E}^{*k}\left[\Lambda_\beta\right] XX^{\dagger}\right].
	\end{split}\end{equation}
	We can now use again the structure of the canonical form, which implies
	\begin{align}
	{\mc E^*}^k [\Lambda_\beta] = \lambda_\beta^{2k} {\mc E_\beta^*}^k[\Lambda_\beta] = \lambda_\beta^{2k} \Lambda_\beta.
	\end{align}
	Thus, we finally find
	\begin{align}\label{eq_restriction}			
	\Tr\left[\Lambda_\beta XX^{\dagger}\right]
		&\leq \lambda_\beta^{2k} \Tr\left[\Lambda_\beta XX^{\dagger}\right].
	\end{align}
	As we have seen above, for $\beta \neq (1,1)$ we have $\lambda_\beta < 1$. In those cases \eqref{eq_restriction} gives a contradiction unless $P_\beta XX^{\dagger} P_\beta = 0$ (with $P_\beta$ the projector onto $\mathcal{H}_\beta$) since $\Lambda_\beta$ is full rank on $\mathcal{H}_\beta$ and vanishes on the complement. Therefore we can restrict the support and image of $XX^{\dagger}$ to be in $\mathcal{H}_1\otimes\mathcal{H}_1$. This implies that the left singular vectors of $X$, thus its image, are contained in $\mathcal{H}_1\otimes\mathcal{H}_1$.
From Lemma~\ref{lemma:symmetry}, we know that $\mathcal{E}[\mf{s}Y\mf{s}] = \mf{s}\mathcal{E}[Y]\mf{s}$ for all $Y$. Together with $\mf s^2=\1$ and \eqref{proof_eq_op}, we then find that $X^\dagger$ fulfills the same equation as $X$: 
	\begin{equation}\label{proof_eq_op_dagger}
		X^{\dagger}
		= \mathcal{E}^{k-1}\left[\mathcal{E}[X^{\dagger}\mf{s}]\mf{s}\right].
	\end{equation}
	Since, by assumption, the solution of this equation is unique up to a scalar factor we have $X \propto X^{\dagger}$. We can thus assume $X$ to be hermitian:
	\begin{equation}\label{hermiticity}
		X = X^{\dagger}.
	\end{equation}
	Since $X$ is hermitian, its image is the same as its support. Therefore \eqref{eq_restriction} and \eqref{hermiticity} together imply that the image and the support of $X$ is contained in $\mathcal{H}_1\otimes\mathcal{H}_1$. Therefore we can rewrite \eqref{proof_eq_op} as 
	\begin{equation}\label{proof_eq_op_restricted}
		\mathcal{E}_{(1,1)}^{k-1}\left[\mathcal{E}_{(1,1)}[X\mf{s}]\mf{s}\right] = X
	\end{equation}
	with $\mathcal{E}_{(1,1)} = \mc T_1\otimes\mc T_1$ (and $\mf s$ is naturally restricted to $\mc H_1\otimes \mc H_1$).

	{\bfseries Step 3.} 
	The map $\mc T_1$ is primitive. For large values of $k$ we therefore have:
	\begin{equation}\label{big_k_approx}
		\mc T_1^k[Y] = \mc P_{\Lambda_1}[Y] + \mc O(e^{-k}) = P_1\Tr[\Lambda_1 Y] + \mathcal{O}(e^{-k}),
	\end{equation}
	where $\Lambda_1$ is normalised such that $\Tr[\Lambda_1] = 1$. The term $\mc{O}(\e^{-k})$ denotes a matrix of norm upper bounded by $K \e^{-a k}$, where $a,K$ are constants (that depend on the bond-dimension $d_1$). 

%	We can see why \eqref{big_k_approx} is true by considering that $P_1$ is the unique fixed point of $\mf{T}_1$ and $\Lambda_1$ is the unique fixed point of $\mf{T}_1^*$.	
	Applying \eqref{big_k_approx} to \eqref{proof_eq_op_restricted}, we get
	\begin{equation}\label{proof_eq_approx}
		X = P_{(1,1)}\, \Tr\left[\Lambda_{(1,1)} \mathcal{E}_{(1,1)}[X\mf{s}]\mf{s}\right] + \mathcal{O}(\e^{-k})
	\end{equation}
	with $\Lambda_{(1,1)} = \Lambda_1\otimes\Lambda_1$ and $P_{(1,1)}$ denotes the orthogonal projector onto $\mathcal{H}_1\otimes\mathcal{H}_1$. 	Imposing $X$ to be of norm $1$ and taking norms on both sides of \eqref{proof_eq_approx} we get
	\begin{equation}\label{proof_ineq_**}
		1 - \mathcal{O}(e^{-k}) \leq \left|\Tr\left[\Lambda_{(1,1)} \mathcal{E}_{(1,1)}[X\mf{s}]\mf{s}\right]\right|\leq 1 + \mathcal{O}(e^{-k}).
	\end{equation}
	Thus from $\eqref{proof_eq_approx}$ we find 
	\begin{equation}\label{X_approx}
		X = \pm P_{(1,1)} + \mathcal{O}(e^{-k}).
	\end{equation}
	Since we can change $X$ by a factor of $\pm 1$, we can deliberately choose 
	the positive sign.
	
	{\bfseries Step 4.} We can now re-insert \eqref{X_approx} into \eqref{proof_ineq_**}, and, using $\mathcal{E}_{(1,1)}[P_{(1,1)}Y] = \mathcal{E}_{(1,1)}[Y]$, we obtain
	\begin{equation}\label{pre_lim_k}
		\Tr\left[\Lambda_{(1,1)} \mathcal{E}_{(1,1)}[\mf{s}]\mf{s}\right]= 1\pm \mathcal{O}(e^{-k}).
	\end{equation}
	While $X$ depends on $k$, $\Tr\left[\Lambda_{(1,1)} \mathcal{E}_{(1,1)}[\mf{s}]\mf{s}\right]$ is independent of $k$ and for all $k_0\in \mathbb{N}$ there exists $k\geq k_0$ that allows us to get to \eqref{pre_lim_k} by assumption. We can therefore take the limit $k\rightarrow\infty$ and find
	\begin{equation}\label{proof_eq_***}
		\Tr\left[\Lambda_{(1,1)} \mathcal{E}_{(1,1)}[\mf{s}]\mf{s}\right] = 1~.
	\end{equation}
	Since $\Lambda_{(1,1)}>0$ on $\mc H_{(1,1)}$, Lemma~\ref{lemma:1D} then implies $\dim(\mc H_1)=1$. 
	
	{\bfseries Step 5.} Since $\mathcal{H}_1\otimes\mathcal{H}_1$ is one-dimensional, we also have $X = X|_{\mc H_1\otimes \mc H_1}= P_{(1,1)}$ and $X$ is also a fixed point of $\mathcal{E}$, which in vectorial representation can be written as 
	\begin{equation}\label{t_x_eq}
		(T\otimes T)\ket{x} = \ket{x}.
	\end{equation}
	Simultaneously we also have $\mc U[X] = X\mf{s} = X$, which means 
	\begin{equation}\label{u_x_eq}
		U\ket{x} = \ket{x}.
	\end{equation}
	By \eqref{t_x_eq} we necessarily have $\ket{x} = \ket{t_1}\ket{t_1}$. Then using the Schmidt decomposition we have 
	\begin{equation}\label{decomp_t1}
		\ket{t_1} = \sum_{i}\alpha_i\ket{\chi_i}\ket{\eta_i}
	\end{equation}
	with $\sum_i \alpha_i^2 = 1$ and $\alpha_i \geq 0$. This implies 
	\begin{align} \label{t1t1_U_img}
		\bra{t_1 t_1}U\ket{t_1 t_1} 
			& = \sum_i \alpha_i^4.
	\end{align}
	From \eqref{u_x_eq} we also have $\bra{t_1 t_1}U\ket{t_1 t_1} = 1$. But then we have
	\begin{equation}\label{product_eq_final}
		1 = \sum_i \alpha_i^2 = \sum_i \alpha_i^4~.
	\end{equation}
	The only way \eqref{product_eq_final} can be satisfied with $\alpha_i \geq 0$ is if there is a unique $j$ such that $\alpha_j = 1$ and $\alpha_i = 0$ for $i\neq j$. But then by \eqref{decomp_t1} we have $\ket{t_1} = \ket{\chi_j}\ket{\eta_j}$, thus obtaining statement $(\textit{1.})$ and completing the proof of the proposition.
\end{proof}

\subsubsection{Equivalence of statements \ref{thm1:1} and \ref{thm1:4}}
\label{sec:1to4}
We now show that the entropy of any finite region vanishes if and only if the highest eigenvector of $T$ is a product-state.
Let us therefore, for any fixed collection of sites $X\subset \mathbbm N$, define the reduced density matrix as
\begin{align}
	\rho^{(n)}_X = \tr_{X^c}[\rho^{(n)}],\quad \forall n\geq |X|. 
\end{align}
For our next step we are interested in the case where $X=\{k\}$ contains only the $k$-th site of the spin chain. Then we have by the cyclic properties of (partial) traces
\begin{widetext}
\begin{equation}\label{tr_rho_single}
	\hat \rho_{\{k\}} =\;\includegraphics[width=0.415\textwidth, valign=c]{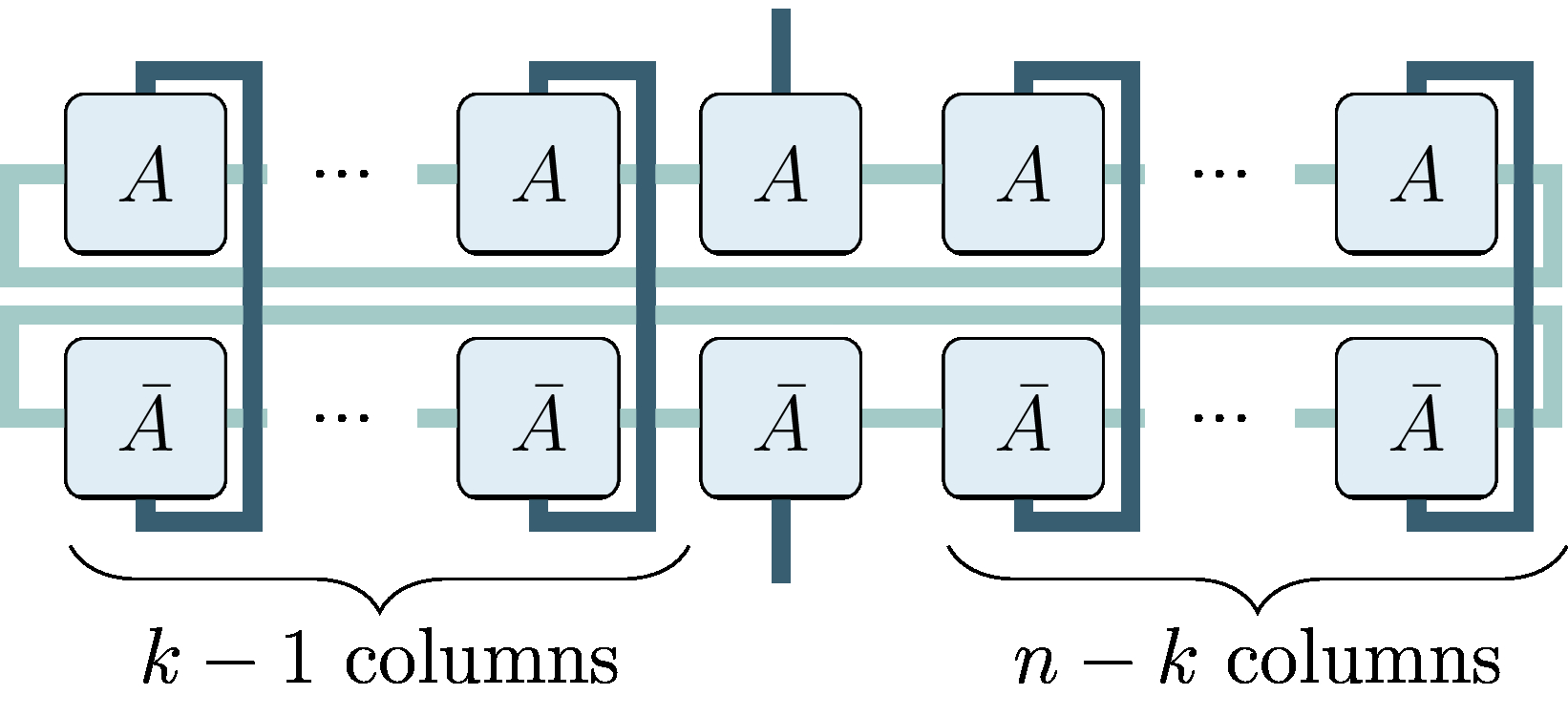}\; =  \; \includegraphics[width=0.252\textwidth, valign=c]{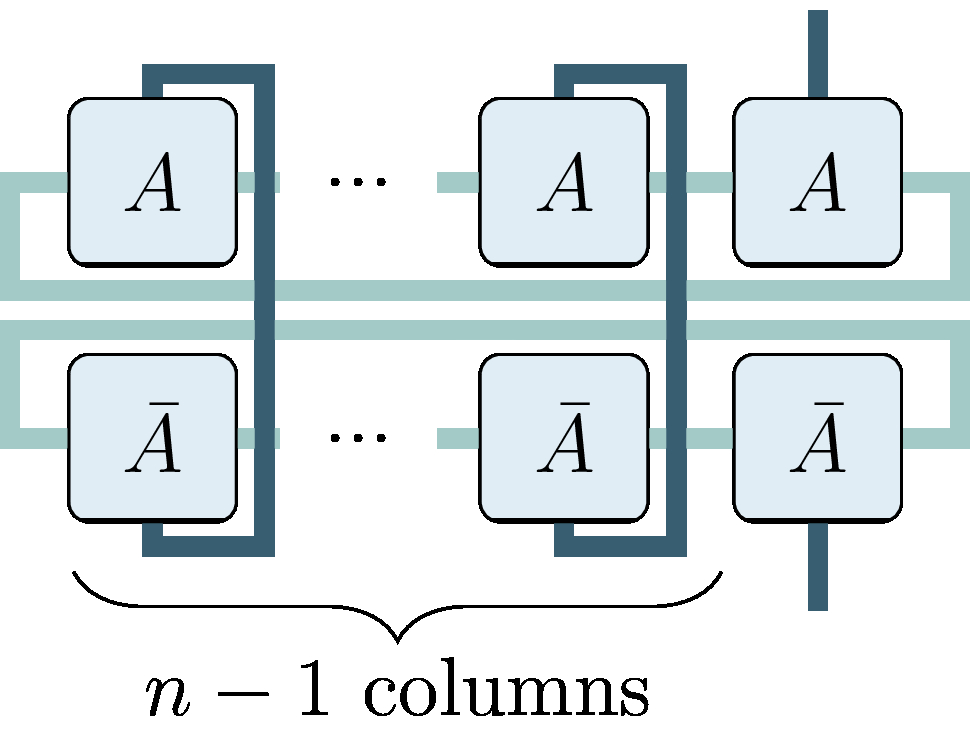}
\end{equation}
\end{widetext}
where we moved the $n-k$ columns from the right to the left without breaking any line.
Then it is clear that (omitting normalization in the diagrams)
\begin{align}\label{s2_single_That}
	\Tr[\hat \rho_{\{k\}}^2] &= \includegraphics[width=0.24\textwidth, valign=c, raise=-11pt]{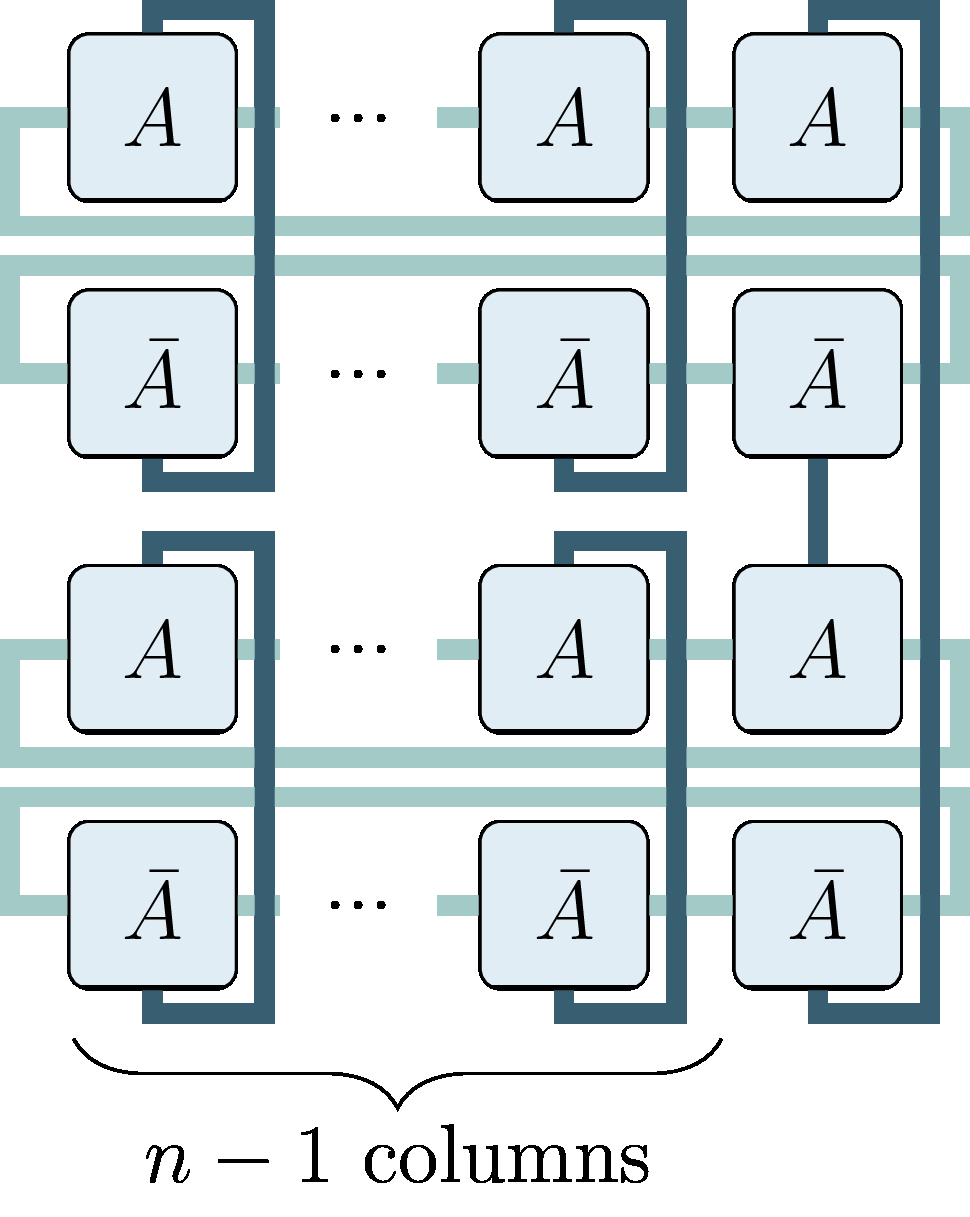}
=\includegraphics[width=0.125\textwidth, valign=c]{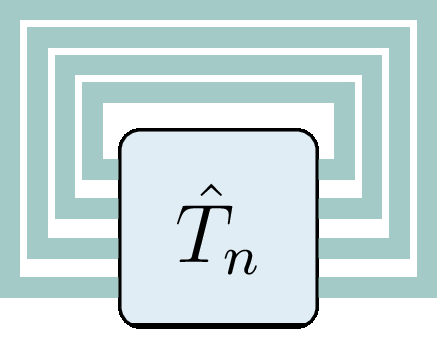}\nonumber \\
	&= \frac{\Tr[\hat T_n]}{\Tr[T^{n}]^2}.
\end{align}
This last insight leads to the following proposition:
\begin{proposition}\label{pure_site}
	Let $\rho^{(n)}$ be the density operator of a spin chain of length $n$ described by a translational invariant MPS with gapped transfer operator $T$. Then 
	\begin{equation}\label{pure_site_s2}
		\lim\limits_{n\rightarrow\infty}\s_2\left(\rho^{(n)}_{\{k\}}\right) = 0
	\end{equation}
	if and only if the greatest eigenvector of $T$ is a product state.
\end{proposition}
\begin{proof}
	By \eqref{s2_single_That}
	\begin{equation}
		\s_2(\rho^{(n)}_{\{k\}}) = -\log(\Tr[(\rho^{(n)}_{\{k\}})^2]) = -\log\left(\frac{\Tr[\hat T_n]}{\Tr[T^n]^2} \right)~.
	\end{equation}
	Since $T$ is gapped we can assume without loss of generality that $t_1 = 1$ and $|t_2| < 1$. By Lemma~\ref{trace_ineq}, we therefore get
	\begin{equation}\label{pure_site_s2_tr}
		\lim_{n\rightarrow \infty} \s_2\left(\rho_{\{k\}}^{(n)}\right) =
		\lim_{n\rightarrow \infty} -\log\Tr[\hat T_n]~.
	\end{equation}
	As in the proof of ($\textit{2.} \Rightarrow \textit{1.}$) of Proposition~\ref{eq_rel_tech} we can approximate $T^n$ for large values of $n$ by the Jordan block corresponding to its largest eigenvalue, we get
	\begin{equation}\label{Tn_approx}
		T^n = VJ^n V^{-1} = \ket{t_1}\bra{e_1}V^{-1} + \mathcal{O}(e^{-n})
	\end{equation}
	with $\ket{t_1}$ the highest eigenvector of $T$, $\ket{e_1}$ the first vector of the canonical basis and $V$ the basis transformation matrix.
	Therefore we have
	\begin{equation}\label{Tn_lim}
		\lim_{n\rightarrow\infty} \hat T_n =  \ket{t_1t_1}\bra{e_1e_1}\left(V^{-1}\otimes V^{-1}\right) U  \left(T\otimes T\right) U.
	\end{equation}
	Which makes it clear that  $\lim_{n\rightarrow\infty} \hat T_n$ has only the top row non-zero in a basis extending $\ket{t_1t_1}$ to an orthonormal basis, thus making it triangular in this basis. Triangular matrices have the eigenvalues on the diagonal, therefore
	\begin{equation}\label{pure_site_s2_formula}
		\lim\limits_{n\rightarrow\infty}\s_2(\rho^{(n)}_{\{k\}}) =  -\log\left(|\hat t|\right),
	\end{equation}
	with 
	\begin{align}\label{eq:t_hat_def}
		\hat t = \bra{e_1e_1}\left(V^{-1}\otimes V^{-1}\right) U  \left(T\otimes T\right) U\ket{t_1t_1}. 
	\end{align}
	
	($\Leftarrow$) By assuming that $\ket{t_1}$ is a product state, $U\ket{t_1}\ket{t_1} = \ket{t_1}\ket{t_1}$ as in \eqref{t1_prod_evU}. Therefore
	\begin{equation}
		\hat t = \bra{e_1e_1}\left(V^{-1}\otimes V^{-1}\right)\ket{t_1t_1} = 1,
	\end{equation}
	since $\ket{e_1} = V^{-1}\ket{t_1}$. 
	
	($\Rightarrow$) By \eqref{eq:t_hat_def}, we have
	\begin{align}
	\lim_{n\rightarrow \infty} \Tr[\hat T_n]= \lim_{n\rightarrow \infty} \bra{t_1t_1}\hat T_n\ket{t_1 t_1}.
	\end{align}
	We now follow a similar strategy as in the proof of the equivalence of statements \ref{thm1:1}--\ref{thm1:3}. Going to Liouville space and making use of the canonical form, we find
	\begin{align}
		\bra{t_1t_1}\hat T_n\ket{t_1 t_1} &= \frac{1}{d_1^2} \Tr\big[P_{(1,1)} \mc E^n[\mc E[P_{(1,1)} \mf s] \mf s\big] \\
		&= \frac{1}{d_1^2} \Tr\big[P_{(1,1)} \mc E_{(1,1)}^n[\mc E_{(1,1)}[\mf s] \mf s\big].
	\end{align}
	Since $\mc E_{(1,1)}$ is a primitive, unital CP-map, we have 
	\begin{align}
		\lim_{n\rightarrow \infty} (\mc E^*_1)^n[P_{(1,1)}/d_1^2] = \Lambda_{(1,1)} = \Lambda_1\otimes \Lambda_1 > 0
	\end{align}
	on $\mc H_1\otimes \mc H_1$. If we now assume that 
	\begin{align}
	\lim_{n\rightarrow\infty} S_2(\rho^{(n)}_{\{k\}}) = 0,
	\end{align}
	we again require
	\begin{align}
	\Tr[\Lambda_{(1,1)} \mc T_1 \otimes \mc T_1 [\mf s]\mf s]=1,
	\end{align}
	We can now argue exactly as in the proof of Proposition~\ref{eq_rel_tech}.
\end{proof}

This last result tells us that, under the assumptions of the proposition, the reduced state of a single site is pure if and only if the greatest eigenvector of the transfer operator is a product state. 
Let us now again consider a fixed, arbitrary, but finite, subsystem $X$ of the chain. Rényi entropies are minimised only by pure states. Therefore, by Proposition~\ref{pure_site}, for any $x\in X$ there exists $\ket{\varphi_x} \in \mathcal{H}_S$ such that
\begin{equation}\label{site_all_lim}
	\lim_{n\rightarrow\infty} \hat \rho^{(n)}_{\{x\}} = \ket{\varphi_x}\bra{\varphi_x}
\end{equation}
if and only if $\ket{t_1}$ is a product state. 
Since this equivalence holds for all $x\in X$, we have that $\ket{t_1}$ being a product state is equivalent to each site of $X$ approaching a pure state in the thermodynamic limit. Therefore
\begin{equation}\label{X_prod_lim}
	\lim_{n\rightarrow\infty} \hat \rho_X^{(n)} = \bigotimes_{x\in X}\ket{\varphi_x}\bra{\varphi_x}
\end{equation}
for some collection $\{\ket{\varphi_x}\}_{x\in X} \subset \mathcal{H}_S$ if and only if the greatest eigenvector of $T$ is a product state. We can note that $\ket{\varphi_x} = \ket{\varphi_y}\;\forall x,y$ because of translational invariance.  
This finshes the proof of Theorem~\ref{thm:eq_rel_null}.

\section{Lower bound to the entropy density}
\label{sec:lower-bounds}
In this section, we provide the full statements and derivation for the lower bound on the entropy density $s^{(k)}_\alpha$ and its implications for quantum channels.  
We work under the same assumptions as in section~\ref{sec:translational-invariant} and assume that the transfer operator $T$ is normalized, so that $t_1=1$, and assume that $T$ is in canonical form. 
Our task is to upper bound the highest eigenvalue of the twisted transfer-operator in the limit $k\rightarrow \infty$.
In Liouville representation, the eigenvalue-equation reads
\begin{align}
	\mc E^{k-1}[\mc E[X_k \mf s]\mf s] = \hat t^{(k)}_1 X_k,
\end{align}
where $\hat t^{(k)}_1$ is the highest eigenvalue of the twisted transfer operator and $X_k$ the corresponding eigenvector, which we can choose to be normalized as $\norm{X_k}=1$.
\begin{lemma}\label{lemma:kinfty}
	In the limit $k\rightarrow \infty$, the highest eigenvector of the twisted transfer-operator $\hat{\mc T}_k$ is given by the projector $P_{(1,1)}$ and the associated eigenvalue is
	\begin{align}
		\hat t_1^{(\infty)} = \Tr[\Lambda_{(1,1)} \mc E[\mf s] \mf s] = \Tr[\Lambda_{(1,1)} \mc T_1\otimes \mc T_1 [\mf s]\mf s].
	\end{align}
	\begin{proof}
	By Lemma~\ref{lemma:powersprojection}, we have
		\begin{align}
		\lim_{k\rightarrow \infty} \mc E^{k-1} = \mc P_{\Lambda_1}\otimes \mc P_{\Lambda_1}. 
		\end{align}
		Thus, in this limit, the eigenvalue equation reduces to 
		\begin{align}
			\hat t^{(\infty)}_1 X_k = P_{(1,1)}\tr[\Lambda_{(1,1)}\mc E[X\mf s]\mf s]=P_{(1,1)}\tr[\Lambda_{(1,1)}\mc E_{(1,1)}[X\mf s]\mf s],\nonumber
		\end{align}
		which gives $X=P_{(1,1)}$ (by normalization) and $\hat t^{(\infty)}_1 = \tr[\Lambda_{(1,1)} \mc E_{(1,1)}[\mf s]\mf s]$, since $\mc E_{(1,1)}[P_{(1,1)}\mf s]=\mc E_{(1,1)}[\mf s]$.
	\end{proof}
\end{lemma}
Having calculated the eigenvalue of the twisted transfer operator in the limit $k\rightarrow \infty$, we are left with the task to estimate it.
Since the only information we have about $\mc E_{(1,1)}$ is that it is a tensor product of primitive quantum channels, we are left with a general problem for primitive quantum channels. We therefore now drop the sub-scripts $(1,1)$ and write $\tilde \Lambda = \Lambda\otimes \Lambda\equiv \Lambda_{(1,1)}$.
In the following, let us further write
\begin{align}
	\hat t(\mc T) := \Tr[(\Lambda\otimes\Lambda)(\mc T\otimes \mc T)[\mf s]\mf s]\geq 0
\end{align}
for any unital CP-map $\mc T$ whose dual $\mc T^*$ has fixed-point $\Lambda$. 

The following theorem provides two bounds for $\hat t(\mc T)$ in the case $\Lambda>0$. The first is the second main result stated in the introduction. The second provides a more involved bound, which replaces the sum over singular values of $\mc T$ with properties of the fixed point and the expansion coefficient. In the general case, when $\mc T^*$ does not have a full-rank fixed point, we see from Lemma~\ref{lemma:kinfty} that we can simply insert  $\mc T=\mc T_1$, $\Lambda=\Lambda_{(1,1)}$ and $D=\dim(\mc H_1)$ to obtain the lower-bound on $\lim_{k\rightarrow\infty} s_\alpha^{k}$. 
\begin{theorem}\label{thm:expansion-that}
	Let $\mc T$ be a primitive, unital CP-map on a $D$-dimensional Hilbert-space. Let $\Lambda>0$ be the fixed-point of $\mc T^*$, with smallest eigenvalue $\lambda_{\min{}}$.  Furthermore, let $d$ be the Kraus-rank of $\mc T$. Then
%	\begin{align}
%		\hat t(\mc T) &\leq 1- \lambda_{\min{}}D^2\Big[\Tr[XXX]+1 - \e^{-S_2(\Lambda)}/D\Big]
%	\end{align}
%	and
	\begin{align}\label{eq:boundthat1}
		\frac{1}{d}\leq \hat t(\mc T) &\leq 1 - \lambda_{\min{}}^2\Big[D^2 - \sum_i s_i(\mc T)^2\Big]
	\end{align}
	and
	\begin{align}\label{eq:boundthat2}
		\frac{1}{d}\leq \hat t(\mc T) &\leq 1 - \lambda_{\min{}}^2D^2\Big[1- \e^{-S_2(\Lambda)}/D \nonumber \\
		&\quad  -\kappa(\mc T)\big[\kappa(\mc T) + 2(\e^{-S_2(\Lambda)}/D- 1/D^2) \big]\Big].
	\end{align}
	The bound~\eqref{eq:boundthat1} is trivial if and only if $D=1$. 
\end{theorem}
The first bound \eqref{eq:boundthat1} makes more detailed use of the singular value distribution of $\mc T$. 
The second bound \eqref{eq:boundthat2} has the advantage of only depending on the expansion coefficient of $\mc T$, but requires more knowledge about the fixed-point $\Lambda$.
\begin{corollary}
	If $\Lambda$ is maximally mixed ($\mc T$ trace-preserving), then
	\begin{align}
		\frac{1}{d}\leq \hat t(\mc T) &\leq  \frac{\sum_i s_i(\mc T)^2}{D^2}\leq \frac{1}{D^2} + \kappa(\mc T)^2. 
	\end{align}
\begin{proof}
	The first bound follows from $\lambda_{\min{}}=1/D$. For the second bound, notice that $s_1(\mc T)=1$ \cite{PerezGarcia2006}, since $\mc T$ is unital and trace-preserving and $\kappa(\mc T) = s_2(\mc T)$. We then have
	\begin{align}
		1- \lambda_{\min{}}^2 \left[D^2 - \sum_i s_i^2(\mc T)^2\right] &= \frac{\sum_i s_i(\mc T)^2}{D^2} \leq \frac{1}{D^2} + s_2(\mc T)^2 \nonumber \\
		&=\frac{1}{D^2} + \kappa(\mc T)^2. 
	\end{align}
\end{proof}
\end{corollary}

Before we get to the details of the proofs of the bounds in theorem~\ref{thm:expansion-that}, let us introduce some notation. 
We define $\mc Q = \mc T - \mc P_\Lambda$ and write  $\mc E \equiv \mc T\otimes \mc T$, $\mc R \equiv \mc E -\mc P_\Lambda\otimes \mc P_\Lambda$, and $\tilde Q:=\mc Q\otimes \mc Q$.
With this notation, the expansion coefficient is simply given by
\begin{align}
	\kappa(\mc T) = \sup_{X:\norm{X}_2=1}\norm{\mc Q[X]}_2 = \norm{\mc Q}=s_1(\mc Q)
\end{align}
and $s_1(\tilde{\mc Q}) = \kappa(\mc T)^2$. 
Apart from this, we have little information about the given maps. However we know that $\mc P_\Lambda[\1]=\1, \mc P^*_\Lambda[\Lambda]=\Lambda$ and $\mc R[\1]=0, \mc R^*[\Lambda\otimes \Lambda]=0$, since $\mc Q[\1]=0$ and $\mc Q^*[\Lambda]=0$. We will also make use of the following Lemmata:
\begin{lemma}\label{lemma:simplebound}
	Let $A,B$ be two hermitian operators with $A\geq a_0\1$. Then
	\begin{align}
	\Tr[AB] \leq \norm{B} \Tr[A-a_0\1] + a_0\tr[B].
	\end{align}
	\begin{proof}
	By assumption, $A-a_0\1$ is a positive semi-definite operator. Similarly, $\norm{B}\1-B\geq 0$. Hence $\Tr[(A-a_0\1)(\norm{B}\1 - B)]\geq 0$. 
	\end{proof}
\end{lemma}
\begin{lemma}\label{lemma:PQ} Let $\mc S$ be a super-operator. Then
	$\mc P_\Lambda\otimes \mc S[\mf s]= \1\otimes \mc S[\Lambda].$
	\begin{proof}
		Let $\Lambda= \sum_i q_i \proj{i}$ be the spectral decomposition of $\Lambda$ and write $\mf s = \sum_{i,j} \ketbra{i}{j}\otimes\ketbra{j}{i}$.
		Then
		\begin{align}
			\mc P_\Lambda\otimes \mc S[\mf s]= \sum_{i,j} \1\otimes \mc S[\ketbra{j}{i}] \Tr[\Lambda\ketbra{i}{j}] = \1\otimes \mc S[\Lambda].\nonumber
		\end{align}
	\end{proof}
\end{lemma}

For both bounds, we will furthermore decompose $\Lambda$ into a convex combination of the maximally mixed state state and a state $\mathbbm B[\Lambda]$ on the boundary of the convex set of $D$-dimensional quantum states (see Fig.~\ref{fig:boundary-state}):
\begin{align}
	\Lambda = (1-x) \frac{\1}{D}  + x \mathbbm B[\Lambda].
\end{align}
It follows that $0\leq x < 1$, since $x=1$ if and only if $\mathbbm B[\Lambda]=\Lambda$, but the rank of $\mathbbm B[\Lambda]$ is strictly smaller than that of $\Lambda$.
A different way to express $x$ is through the minimum eigenvalue of $\Lambda$, as $x = 1-\lambda_{\min{}} D$.
\begin{figure}\centering
	\includegraphics[width=5cm]{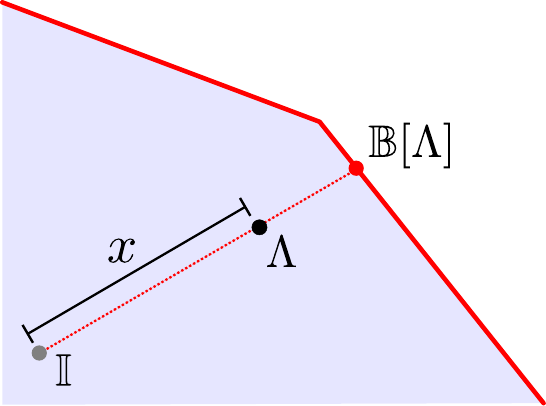}
	\caption{Convex decomposition of the fixed-point $\Lambda$ into the maximally mixed state $\mathbbm{I}=\1/D$ and a state on the boundary of the $D$-dimensional state-space.}
	\label{fig:boundary-state}
\end{figure}
	It is useful to write $X:= \Lambda-\1/D$, which leads to $\Lambda = \1/D + x X$ and $\Tr[X]=0$ as well as $x^2\Tr[X^2]=\Tr[\Lambda^2]-1/D$. 
	
	Similarly, it will be useful to introduce $\tilde X = \Lambda\otimes \Lambda - \1/D^2$ so that $\Lambda\otimes \Lambda = \1/D^2 + \tilde X$. Then $\Tr[\tilde X]=0$ as well as
	\begin{align}
	\Tr[\tilde X^2] = \Tr[\Lambda^2]^2 - 1/D^2
	\end{align}
	and
	\begin{align}\label{eq:lowerboundtildeX}
	\tilde X \geq - \frac{x(2-x)}{D^2}\1.
	\end{align}

	With these preparations, we can now give the proof of the theorem. 
\begin{proof}[Proof of the lower bound $1/d\leq \hat t(\mc T)$] 
	Let $\mc T[\cdot]= \sum_{i=1}^d A_i \cdot A_i^\dagger$ be a minimal Kraus decomposition of $\mc T$ and use the $A_i$ to construct an MPS with physical dimension $d$. Then the R\'enyi-2 entanglement density of every $k$-th spin in the thermodynamic limit is upper bounded by $\log(d)$ and approaches $-\log(\Tr[\Lambda\otimes\Lambda \mf s \mc T\otimes \mc T[\mf s]])$ in the limit $k\rightarrow \infty$ by Lemma~\ref{lemma:kinfty}. Hence $\log(d) \geq  - \log(\hat t(\mc T))$, which proves the inequality. 
\end{proof}

\begin{proof}[Proof of upper bound in \eqref{eq:boundthat1}]
	We have
	\begin{align}
		\hat t(\mc T) &= \tr[\mf s\mc T\otimes \mc T[\mf s]]/D^2 + \tr[\tilde X\mf s\mc T\otimes \mc T[\mf s]]\\
		&= \frac{\sum s_i(\mc T)^2}{D^2} +  \tr[\tilde X\mf s\mc T\otimes \mc T[\mf s]],
	\end{align}
	where we used Lemma~\ref{lemma:swap-singular-values}. Using \eqref{eq:lowerboundtildeX}, $\tr[\tilde X]=0$, $\norm{\mf s\mc T\otimes \mc T[\mf s]}\leq 1$, and Lemma~\ref{lemma:simplebound} we then find
	\begin{align}
		\tr[\tilde X\mf s\mc T\otimes \mc T[\mf s]] &\leq \frac{x(2-x)}{D^2}\left(\tr[ \1 ] - \tr[\mf s\mc T\otimes \mc T[\mf s]]\right)\nonumber\\
		&=\frac{x(2-x)}{D^2}\left(D^2 - \sum_i s_i(\mc T)^2\right).\nonumber
	\end{align}
	Using $x(2-x) = 1 -(1-x)^2$ and $(1-x)^2/D^2 = \lambda_{\min{}}^2$ then yields the desired bound.  
\end{proof}
	\begin{proof}[Proof of upper bound in \eqref{eq:boundthat2}] 
	We first decompose the map $\mc T$ as $\mc T=\mc P_\Lambda+\mc Q$. Using $\Tr[\Lambda^2]=\Tr[\Lambda\otimes\Lambda\mf s]$, one then finds
	\begin{align}
		\hat t(\mc T) &= \Tr[\Lambda^2]^2 + 2\Tr[\Lambda\otimes \Lambda \mc P_\Lambda \otimes \mc Q[\mf s]\mf s] + \Tr[\Lambda\otimes\Lambda\mf s\tilde{\mc Q}[\mf s]]\nonumber \\
		     &= \Tr[\Lambda^2]^2 + 2\Tr[\Lambda^2 \mc Q[\Lambda]] + \Tr[\Lambda\otimes\Lambda \mf s\tilde{\mc Q}[\mf s]], \label{eq:firstboundthat}
	\end{align}
	where we used Lemma~\ref{lemma:PQ} to write 
	\begin{align}
		\Tr[\Lambda\otimes \Lambda \mc P_\Lambda \otimes \mc Q[\mf s]\mf s] &= \Tr[(\Lambda\otimes \Lambda)(\1 \otimes \mc Q[\Lambda])\mf s]\nonumber\\ 
		&=\Tr[\Lambda^2 \mc Q[\Lambda]].\nonumber
	\end{align}
	We now first reformulate the second term in~\eqref{eq:firstboundthat}.
	Using the decomposition $\Lambda=\1/D+xX$ together with $\mc Q[\1]=0=\mc Q^*[\Lambda]$, we find
	\begin{align}
		\Tr[\Lambda^2\mc Q[\Lambda]]&= x^2\Tr[X\mc Q[X]]/D + x^3\Tr[X^2\mc Q[X]]\nonumber \\
		&= x^2\Tr[X\mc Q[X]]/D + x^3\Tr[X^2\mc T[X]] \nonumber\\
		&\quad- x^3\Tr[X^2]\Tr[\Lambda X]\nonumber\\
		&= x^2\Tr[X\mc Q[X]]/D + x^3\Tr[X^2\mc T[X]]\nonumber\\
		&\quad-(\Tr[\Lambda^2]-1/D)^2.\nonumber
	\end{align}
	Here, we used $\mc Q = \mc T-\mc P_\Lambda$ and (using $\tr[X]=0$)
	\begin{align}
		x\, \Tr[\Lambda X] &= x\Tr[(\1/D+xX)X] = x^2\,\Tr[X^2]\nonumber \\
		&=\Tr[\Lambda^2]-1/D.
	\end{align}
	Let us now rewrite the term $\Tr[\Lambda\otimes \Lambda \mf s\tilde{\mc Q}[\mf s]]$  using $\tilde X$:
	\begin{align}
		\Tr[\Lambda\otimes\Lambda \mf s\tilde{\mc Q}[\mf s]]
		&=	\Tr[\mf s\tilde{\mc Q}[\mf s]]/D^2+	\Tr[\tilde X \mf s\tilde{\mc Q}[\mf s]]\nonumber\\
		&=	\Tr[\mf s\tilde{\mc Q}[\mf s]]/D^2+	\Tr[\tilde X \mf s{\mc E}[\mf s]]\nonumber\\ 
		&\quad -2 \Tr[\tilde X \mf s (\mc T[\Lambda]\otimes \1)]+	\Tr[\tilde X \mf s]\Tr[\Lambda^2]\nonumber,
	\end{align}
	where we used $\mc T\otimes \mc P_\Lambda[\mf s]=\mc T[\Lambda]\otimes \1$ and $\mc P_\Lambda\otimes \mc P_\Lambda[\mf s] = \1\Tr[\Lambda^2]$.
	Using $\tilde X = (\1/D + xX)^{\otimes 2} - \1/D^2$ and $\Lambda = \1/D+xX$ repeatedly, we find
	\begin{align}
		\Tr[\tilde X \mf s (\mc T[\Lambda]\otimes \1)]&= (\Tr[\Lambda^2]-1/D)/D + x^3\Tr[X^2\mc T[X]]\nonumber \\
		&\quad+2x^2\Tr[X \mc Q[X]]/D.
	\end{align}
	Inserting the derived expression into  \eqref{eq:firstboundthat}, we obtain
	\begin{multline}
		\hat t(\mc T) = \Tr[\Lambda^2]/D + \Tr[\mf s\tilde{\mc Q}[\mf s]]/D^2 + \Tr[\tilde X\mf s\mc E[\mf s]] \\
		- 2x^2\tr[X\mc Q[X]]/D~.\label{eq:secondexpthat}
	\end{multline}
	We can now use Lemma~\ref{lemma:simplebound} and Lemma~\ref{lemma:sEsnorm} together with \eqref{eq:lowerboundtildeX} to find
	\begin{align}\label{eq:boundtildeXsEs}
	\Tr[\tilde X\mf s\mc E[\mf s]]\leq x(2-x)(1- \Tr[\mf s\mc E[\mf s]]/D^2).
	\end{align}
	Let us now discuss the term 
	\begin{align}
		\Tr[\mf s\mc E[\mf s]]/D^2 = \frac{\sum_i s_i(\mc T)^2}{D^2} \leq 1, 
	\end{align}
	where we used Lemma~\ref{lemma:swap-singular-values}. 
	We can again decompose $\mc E$ to get
	\begin{align}
		\Tr[\mf s\mc E[\mf s]]/D^2 =	\Tr[\mf s\tilde{\mc Q}[\mf s]]/D^2 +2\Tr[\mc T[\Lambda]]/D^2 - \Tr[\Lambda^2]/D.\nonumber   
	\end{align}
	Decomposing $\Lambda=\1/D + xX$ and using $\mc Q^*[\Lambda]=0$ then yields
	\begin{align}
	\tr[\mc T[\Lambda]]/D^2 = \Tr[\Lambda^2]/D-x^2 \Tr[X\mc Q[X]]/D. 
	\end{align}
	Plugging the result together with \eqref{eq:boundtildeXsEs} into \eqref{eq:secondexpthat} finally gives
	\begin{align}
	\hat t(\mc T) \leq 1 - (1-x)^2&\left[1-\Tr[\Lambda^2]/D+2x^2\tr[X\mc Q[X]]/D\right.\nonumber\\
		&\quad\quad\quad\left.-\Tr[\mf s\tilde{\mc Q}[\mf s]]/D^2\right]. \label{eq:thirdboundthat}
	\end{align}
	We now make use of the expansion coefficient to bound
	\begin{align}
		x^2\tr[X\mc Q[X]] &\leq x^2\Tr[X^2]\kappa(\mc T)=(\Tr[\Lambda^2]-1/D)\kappa(\mc T),\nonumber \\
		\tr[\mf s\tilde{\mc Q}[\mf s]]/D^2 &\leq \frac{\Tr[\mf s^2]}{D^2}\kappa(\mc T)^2.\nonumber
	\end{align}
	Plugging these bounds into~\eqref{eq:thirdboundthat} and using $1-x = \lambda_{\min{}}D$ together with $\Tr[\Lambda^2]=\exp(-S_2(\Lambda))$, we obtain the final expression:
	\begin{align}
		\hat t(\mc T) \leq  1 &- \lambda_{\min{}}^2D^2\Big[1- \e^{-S_2(\Lambda)}/D \nonumber \\
		 &-\kappa(\mc T)\big[\kappa(\mc T) + 2(\e^{-S_2(\Lambda)}/D- 1/D^2) \big]\Big].
	\end{align}
\end{proof}

\section{Discussion and open problems}
We comprehensively studied the R\'enyi entanglement entropies of disconnected subsystems in one-dimensional spin chains described by translational invariant matrix product states with gapped transfer operator and fixed bond-dimension. 
Our results show that these entropies are either extensive or the system's state approaches a product state in the thermodynamic limit. 
We can furthermore give explicit lower bounds to the entropy density, which only depends on the expansion coefficient and the fixed point of the dual to the transfer operator.
As a side-result we find a simple bound on the expansion coefficient and the distribution of singular values for primitive, unital and trace-preserving CP-maps in terms of their Kraus rank.  

Our work suggests several natural open problems to be considered in future work:
First, Besides the translation invariant case, it would be interesting to study the case where each of the MPS tensors for the different sites is drawn independently at random. One then expects that with very high probability the entanglement entropy is quite large, i.e., close to saturating the upper bound $|\partial A|\log(D)$. It would be interesting to give quantitative bounds on the typical behaviour of the entanglement entropy for disconnected sub-systems, extending the results known for fixed, small sub-systems \cite{Collins2013}. We also assumed a fixed bond dimension and it would be interesting to generalize our results to the case $D=\mathrm{poly}(n)$. 

Secondly, an obvious further generalization to consider is the two-dimensional case, described by Projected Entangled Pairs States (PEPS). Due to the lack of a similarly powerful, local canonical form of the tensors describing the state, we expect this to be much more challenging than the one-dimensional case studied in this paper, but also potentially giving rise to richer behaviour.

Finally, it's not clear to us how tight our bound on the expansion coefficient of primitive quantum channels is away from the trivial limit $\kappa=0$. Numerics with randomly sampled channels as well as the example in the introduction indicate that the bound is not very tight. However, at the moment we cannot rule out fine-tuned examples that come close to saturating our bound. We leave it as an open question to study the tightness of the bound in more detail. Similarly, we leave it as an open question to develop tighter bounds on the averaged squared singular values $\sum_i s_i(\mc T)^2/D^2$ of a primitive quantum channel, which would directly imply better bounds on the entropy density.

\section*{Acknowledgements}
We acknowledge support from the National Science Foundation through  SNSF project No.\ $200020\_165843$ and through the National Centre of
Competence in Research \emph{Quantum Science and Technology} (QSIT). 

%%%%%%%%%%%%
%\bibliographystyle{stdWithTitle}
\bibliographystyle{apsrev4-1}
\bibliography{literature.bib}
%%%%%%%%%%%%%
\newpage
%%%%%%%%%%%%%%%
%            Appendix
%%%%%%%%%%%%%%%

\appendix
\noindent{\LARGE{\sffamily {Appendix}}}

\section{The MPS constructed from $\mc T'$}\label{app:exampleMPS}
We here briefly describe the construction and properties of the MPS defined using the channel $\mc T'$. The channel has two Kraus-operators $U_1/\sqrt{2}$ and $U_2/\sqrt{2}$. 
The only properties of $\mc T'$ that we will make use of are: i) that $\tr[U_1^\dagger U_2]=0$ and ii) that $\mc T'$ squares to the constant map, mapping any state to the maximally mixed state: 
\begin{align}
	{\mc T'}^2 = \mc P_{\1/2}.
\end{align}
Since $\mc T'[\1]=\1$, we have $\mc P_{\1/2} \circ \mc T' = \mc T'\circ \mc P_{\1/2} = \mc P_{\1/2}$ and therefore ${\mc T'}^{n+1} = \mc P_{\1/2}$ for any $n\geq 1$. 
\begin{lemma}\label{lemma:MPSexample}
	Let $\mc T$ be a unital or trace-preserving CP-map on a $D$-dimensional Hilbert-space such that $\mc T^{k}=\mc P_{\1/D}$ for some $k\in\mathbbm N$. Then $\mc T$ has spectrum (including multiplicities) $(1,0,\ldots,0)$ and is both unital and trace-preserving. 
	\begin{proof}
		Let $X$ be an eigenvector of $\mc T$ with eigenvalue $\lambda$. Then it is also an eigenvector of $\mc T^k=\mc P_{\1/D}$ with eigenvalue $\lambda^k$. But $\mc P_{\1/D}$ has spectrum $(1,0,\ldots,0)$ and the only eigenvector is given by $\1$. Therefore $\1$ is also the unique eigenvector of $\mc T$. Since $\mc P_{\1/D}=\mc P_{\1/D}^*$, the same conclusion holds for $\mc T^*$.  
	\end{proof}
\end{lemma}
We will also need to compute $\tr[\mf s \mc T'\otimes \mc T'[\mf s]]$ and $\tr[\mf s\mc P_{\1/2}\otimes \mc P_{\1/2}[\mf s]]$. 
By the same computation as in the proof of lemma~\ref{lemma:swap-singular-values} together with $T' = \frac{1}{2}\sum_i U_i\otimes \overline{U_i}$ being the matrix representation of $\mc T'$ one finds
\begin{align}
	\tr[\mf s \mc T'\otimes \mc T'[\mf s]] &= \frac{1}{4}\sum_{i,j} |\tr[U_i U_j^\dagger]|^2=2, 
\end{align}
where we used $\tr[U_1^\dagger U_2]=0$ and $\tr[U_i^\dagger U_i]=\tr[\1]=2$. Lemma~\ref{lemma:swap-singular-values} also directly shows
 $\tr[\mf s\mc P_{\1/2}\otimes \mc P_{\1/2}[\mf s]]=1$, since $\mc P_{\1/2}$ is an orthogonal (in Hilbert-Schmidt inner product) projector. 
We can now derive the properties of the MPS constructed from the tensor $A_1 = U_1/\sqrt{2}$ and $A_2=U_2/\sqrt{2}$.
Let us first compute the R\'enyi-2 entropy of a connected subsystem $X$ of size $l\leq n-2$. First note that the state is automatically normalized since $\tr[T'^n]=1$ from lemma~\ref{lemma:MPSexample}.
Then, following the calculations in section~\ref{sec:background}, we find 
\begin{align}
\tr[\rho_X^2] &= \tr[(T'\otimes T')^{n-l} U (T'\otimes T')^l U]\\
	&= \tr[P\otimes P U(T'\otimes T')^l U],	
\end{align}
where $P$ is the matrix-representation of $\mc P_{\1/2}$. 
In Liouville-space, this equation reads
\begin{align}
	\tr[\rho_X^2] = \frac{1}{4}\tr[\mf s (\mc T'\otimes \mc T')^l[\mf s]]. 
\end{align}
For $l=1$, we thus find $\tr[\rho_X^2]=1/2$, leading to $S_2(\rho_X)=\log(2)$. 
If $l\geq 2$, on the other hand, we find 
\begin{align}
\tr[\rho_X^2] = \frac{1}{4}\tr[\mf s (\mc P_{\1/2}\otimes \mc P_{\1/2})[\mf s]] = \frac{1}{4}. 
\end{align}
Therefore, any connected sub-system of size $l\geq 2$ has entropy $S_2(\rho_X)=2\log(2)$, saturating the bound $S_2(\rho_X) \leq |\partial X| \log(2) = 2\log(2)$.

Finally, let us show that any two sub-systems $X$ and $Y$ (of sizes $l,l'$), that are a distance $L\geq 2$ apart from each other, are uncorrelated, $\rho_{XY}=\rho_X\otimes \rho_Y$ (since the system is periodic, $L\geq 2$ also requires $n-l-l'-L\geq 2$). 
This is particularly easy to see when directly calculating the matrix-elements of $\rho_{XY}$ as follows, using that ${T'}^2=P=\proj{t_1}$ with $\ket{t_1}$ the unique eigenvector of $T'$: 
\begin{widetext}
\begin{align}
	\bra{j_1,\ldots,j_{l+l'}}\rho_{XY}\ket{i_1,\ldots,i_{l+l'}}&= \tr[{T'}^{n-l-l'-L}(A_{i_1}\otimes \overline{A}_{j_1}) \cdots (A_{i_l}\otimes \overline{A}_{j_l}) {T'}^L (A_{i_{l+1}}\otimes \overline{A}_{j_{l+1}})\cdots(A_{i_{l+l'}}\otimes \overline{A}_{j_{l+l'}})]\\
	&= \tr[\proj{t_1} (A_{i_1}\otimes \overline{A}_{j_1}) \cdots (A_{i_l}\otimes \overline{A}_{j_l}) \proj{t_1} (A_{i_{l+1}}\otimes \overline{A}_{j_{l+1}})\cdots(A_{i_{l+l'}}\otimes \overline{A}_{j_{l+l'}})] \\
	&= \underbrace{\bra{t_1}(A_{i_1}\otimes \overline{A}_{j_1}) \cdots (A_{i_l}\otimes \overline{A}_{j_l})\ket{t_1}}_{=\bra{j_1,\ldots, j_l} \rho_X\ket{i_1,\ldots, i_l}} 
	\underbrace{\bra{t_1} (A_{i_{l+1}}\otimes \overline{A}_{j_{l+1}})\cdots(A_{i_{l+l'}}\otimes \overline{A}_{j_{l+l'}})]\ket{t_1}}_{=\bra{j_{l+1},\ldots,j_{l+l'}} \rho_Y \ket{i_{l+1},\ldots,i_{l+l'}}}\\
&=\bra{j_1,\ldots, j_l} \rho_X\ket{i_1,\ldots, i_l} \bra{j_{l+1},\ldots,j_{l+l'}} \rho_Y \ket{i_{l+1},\ldots,i_{l+l'}}.
\end{align}
\end{widetext}

\section{Lemmas and their proofs}\label{app:lemmas}
In this section we cover technical lemmas and proofs of lemmas that are not part of the main text.

\begin{lemma}\label{unit_circle_powers}
	Let $\{z_k\}_{k=1}^d$ be a collection of numbers on the complex unit circle. Then for all $\epsilon > 0$ and for any $n\in \mathbb{N}$ there exists a natural number $N \geq n$ such that
	\begin{equation}\label{circle_pow_pres}
		|z_k^N - 1| < \epsilon~
	\end{equation}
	for all $k = 1,...,d$.
\end{lemma}
\begin{proof}
	Since $|z_k| = 1$ for $k=1,...,d$ we can write $z_k^n = e^{2\pi i \theta_k}$ for some real number $\theta_k$. By the simultaneous version of Dirichlet's approximation theorem \cite{Cassels1965} for any $l\in\mathbb{N}$ there exists $p_1,...,p_d,q \in \mathbb{Z}$ such that $1\leq q \leq l$ and
	\begin{equation}\label{dirichlet}
		|\theta_k - \frac{p_k}{q}| \leq \frac{1}{l^{1/d}q}~.
	\end{equation}
	Thus by defining $\delta_k = \theta_k - \frac{p_k}{q}$ we have $z_k^{nq} = e^{2\pi i q\theta_k} = e^{2\pi i (q\delta_k + p_k)} = e^{2\pi i q\delta_k}$. Therefore
	\begin{align}
		\left|z_k^{nq}-1\right| &= \left|\sum_{j=1}^\infty \frac{(2\pi i q\delta_k)^j}{j!}\right|
		\leq \sum_{j=1}^\infty \frac{(2\pi q\left|\delta_k\right|)^j}{j!}\\
		&= e^{2\pi q\left|\delta_k\right|}-1 
		< 4\pi q\left|\delta_k\right| 
		\leq \frac{4\pi}{l^{1/d}}~,
	\end{align}
	where we used $e^x-1<2x$ for $|x| \leq 1$ (which can be achieved here by choosing $l$ large enough) and \eqref{dirichlet} in the last inequality. Therefore by choosing $l \geq \max[(4\pi/\epsilon)^d, 1]$ and $N = nq$ we achieve \eqref{circle_pow_pres} with $N\geq n$.
\end{proof}

\begin{proof}[Proof of Lemma~\ref{trace_ineq}] \label{trace_ineq_proof}
	($\Rightarrow$) We will prove the contrapositive, therefore we first assume $|a_1| < 1$. \\
	Then,
	\begin{align}
		\lim_{n\rightarrow\infty}\left|\Tr[A^n]\right| &= \lim_{n\rightarrow\infty}\left|\sum_{i=1}^{d} a_i^n \right|
		\leq \lim_{n\rightarrow\infty}\sum_{i=1}^{d} |a_i|^n \\
		&\leq \lim_{n\rightarrow\infty} d|a_1|^n = 0.
	\end{align}
	Therefore $\lim_{n\rightarrow\infty}\log(\Tr[A^n]) = -\infty$. 
	Conversely, if we assume $|a_2| = 1$ we define $2 \leq l \leq d$ such that $|a_i| = 1$ for $i=1,...,l$ and if $l < d$ then $|a_{l+1}| < 1$. By Lemma~\ref{unit_circle_powers} we can construct a strictly increasing sequence of integers $\{n_k\}_{k=0}^\infty$ such that $|a_i^{n_k} - 1| < 1/2^k$. Thus it is clear that
	\begin{equation}
		\lim_{k\rightarrow\infty} \sum_{i=1}^l a_i^{n_k} = l~. 
	\end{equation}
	If $l < d$ then it is also clear that $\lim_{n\rightarrow\infty} \sum_{i=l+1}^d a_i^{n} = 0$. Therefore $\lim_{k\rightarrow\infty}\left|\Tr[A^{n_k}]\right| = l \geq 2$. Since $n_{k+1} > n_k$, we have just proven that 
	\begin{equation}
		\limsup_{n\rightarrow\infty} \left|\Tr[A^n]\right| \geq l \geq 2~.
	\end{equation}
	Thus $\limsup_{n\rightarrow\infty} \log\left|\Tr[A^n]\right| \geq \log(2) > 0$, so if $\lim_{n\rightarrow\infty} \log\left|\Tr[A^n]\right|$ exists it is strictly greater than 1, otherwise it is not defined (and therefore cannot be 0).
	
	($\Leftarrow$) We assume that $|a_1| = 1$ and $|a_2| < 1$. Then 
	\begin{align}
		\lim_{n\rightarrow\infty}\left|\Tr[A^n] - a_1^n\right| &= \lim_{n\rightarrow\infty}\left|\sum_{i=2}^{d} a_i^n\right| \leq \lim_{n\rightarrow\infty}\sum_{i=2}^{d} |a_i|^n \nonumber\\
		&\leq \lim_{n\rightarrow\infty} d|a_2|^n  = 0.
	\end{align}
	Therefore $\lim_{n\rightarrow\infty}\log|\Tr[A^n]| = \lim_{n\rightarrow\infty}\log|a_1|^n = \log(1) = 0$, which completes the proof.
\end{proof}

\begin{proof}[Proof of Lemma~\ref{lim_sup}] \label{lim_sup_proof} We make use of complex analysis as presented in a stack-exchange answer \cite{mathstack}.
	Let us define $f:\mathbb{C}\rightarrow\mathbb{C}$
	$$ f(z) = \sum_{i=1}^{n}\frac{1}{1-z\alpha_i}.$$
	We can see that $f$ is a meromorphic function with poles of order $1$ at $1/\alpha_i$ for $i=1,...,n$. The derivatives of $f$ are
	$$ f^{(k)}(z) = k!\sum_{i=1}^{n}\frac{\alpha_i^k}{(1-z\alpha_i)^{k+1}}.$$
	Hence we have $f^{(k)}(0) = k!\sum_{i=1}^{n} \alpha_i^k$ and we can compute the Taylor series of $f$ around $0$ for $|z| < R$, where $R$ is the convergence radius of the Taylor series:
	$$ f(z) = \sum_{k=0}^{\infty}\left(\sum_{i=1}^{n}\alpha_i^k\right)z^k.$$
	Since $f$ is meromorphic its convergence radius around $0$ is $R = \min\limits_{i=1,...,n} \frac{1}{|\alpha_i|}  = 1/r$. Then by the Cauchy-Hadamard theorem
	$$\limsup_{k\rightarrow\infty} \left|\sum_{i=1}^{n} \alpha_i^k \right|^{1/k} = \frac{1}{R} = r~, $$
	proving the lemma.
\end{proof}

\end{document}